\documentclass[11pt,letterpaper]{article}
\usepackage[left=1in, right=1in, top=1in, bottom=1in]{geometry}
\usepackage{booktabs} 
\usepackage[ruled,noend]{algorithm2e} 
\usepackage[utf8]{inputenc}

\SetAlFnt{\small}
\SetAlCapFnt{\small}
\SetAlCapNameFnt{\small}
\SetAlCapHSkip{0pt}
\IncMargin{-\parindent}

\usepackage[]{local}

\author{
Zohar Barak\thanks{Tel Aviv University, \url{zoharbarak@mail.tau.ac.il}} 
\and
Anupam Gupta\thanks{New York University and Google Research, \url{anupam.g@nyu.edu}}
\and
Inbal Talgam-Cohen\thanks{Tel Aviv University, \url{inbaltalgam@gmail.com}}
}
\date{}

\begin{document}


\title{MAC Advice for Facility Location Mechanism Design} 

\maketitle



\begin{abstract}

Algorithms with predictions have attracted much attention in the last years across various domains, including variants of facility location, as a way to surpass traditional worst-case analyses. We study the $k$-facility location mechanism design problem,
where the $n$ agents are strategic and might misreport their
location.

Unlike previous models, where predictions are for the $k$ optimal facility locations, we receive $n$ predictions for the
locations of each of the agents.
However, these predictions are only ``mostly'' and ``approximately''
correct (or $\mac$ for short) --- i.e., some $\delta$-fraction of the predicted locations are
allowed to be arbitrarily incorrect, and the remainder of the
predictions are allowed to be correct up to an $\eps$-error. 
We make no assumption on the independence of the errors.
Can such predictions allow us to beat the current best bounds for strategyproof
facility location?

We show that the $1$-median 
(geometric median)
of a set of points is
naturally robust under corruptions, which leads to an algorithm
for single-facility location with MAC predictions. We extend the robustness result to a ``balanced'' variant of the $k$ facilities case. 
Without balancedness, we show that robustness completely breaks down, even for the setting of $k=2$ facilities on a line.
For this ``unbalanced'' setting, we devise a truthful random mechanism that
outperforms the best known result of Lu et al.~[2010],
which does not use
predictions. En route, we introduce the problem of ``second'' facility
location (when the first facility's location is already fixed).  Our
findings on the robustness of the $1$-median and more generally $k$-medians may
be of independent interest, as quantitative versions of classic breakdown-point results in robust statistics.

\end{abstract}


\section{Introduction}
\label{sec:intro}
Algorithms with predictions is a popular field of study in recent years, falling within the paradigm of beyond the worst case analysis of algorithms (see~\cite{DBLP:journals/cacm/MitzenmacherV22} for a comprehensive survey). 
Motivated by developments in machine learning (ML), this area of study assumes that algorithms can access predictions regarding the input or solution, thus leveraging predictability in typical computational problems. 
Recently, \citet{agrawal2022learning} and \citet{ijcai2022p81} proposed to study predictions in the context of mechanism design, where they have the potential to transform existing designs by adding information about agents' private knowledge (see~\cite{BalkanskiGT23} for a summary of recent research). 

Our focus in this work is on the canonical problem of facility location.
The (utilitarian) $k$-facility location problem is as follows: 
Consider a multi-set of $n$ points $X \subset \mathbb{R}^d$ and $k$ facility locations; the cost of point $x\in X$ is the minimum distance between $x$ and any of the $k$ locations. 
The utilitarian (or social) cost is the sum of costs of all points in $X$. 
The goal is to compute $k$ locations that minimize the utilitarian cost. 
In the context of mechanism design, the points $X$ are the privately-known true locations of strategic agents, and the mechanism receives location reports 
(see \cite{ijcai2021p596} for a recent survey). 
Facility location mechanism design with predictions has been studied by~\citet{agrawal2022learning,ijcai2022p81}.%
\footnote{Predictions have also been studied for the online non-strategic version of facility location --- see Appendix~\ref{sub:more-related}.} 
In this work, we assume a prediction for each point --- that is, the advice $X'$ consists of a prediction $x'_i$ for every $x_i\in X$. 

\paragraph{Worst-case prediction error.} 
In the majority of the literature, the goal of (algorithm or mechanism) design with predictions centers around two seemingly-conflicting properties: \emph{robustness} to erroneous predictions, and \emph{consistency} with predictions that are non-erroneous. 
A fairly common way of achieving both consistency and robustness is interpolating between the approaches of completely following the predictions, and applying the best worst-case algorithm while disregarding the predictions.
An ideal design is one that gracefully transitions, as the \emph{prediction error} increases, from optimal performance when the prediction is correct, to the best-known worst-case performance ~(\cite{lykouris2021competitive,purohit2018improving, agrawal2022learning}). 
By definition, achieving such graceful degradation hinges on how the prediction error is measured.

Arguably the most common measure of prediction error is the distance between the predicted and actual values: 
If $X = \{ x_1, \ldots, x_n \}$ contains the actual values and $X' = \{ x'_1, \ldots, x'_n \}$ is the prediction, the error is defined as $\eta = \ell_p(X - X')$, where $\ell_p$ represents either the $\ell_1$ norm ($\ell_1(t) = \sum_i |t_i|$) or the $\ell_\infty$ norm ($\ell_\infty (t) = \max_i |t_i|$).
This is in some sense a ``worst case'' measure of prediction error, since a single deviation in one of the $n$ entries can significantly inflate the error (the $\ell_1$ and $\ell_\infty$ norms are sensitive to large changes in individual coordinates).
Prediction errors measured in this way appear, e.g., in the context of ski rental (~\cite{purohit2018improving}), online caching ~(\cite{lykouris2021competitive,rohatgi2020near}), the secretary problem ~(\cite{antoniadis2020secretary}), 
single-item auctions~(\cite{ijcai2022p81}), and many other problems.%
\footnote{In some of these works, the prediction is a single value characterizing $X$, which is sensitive to a single change in~$X$.} We address this type of prediction error as the ``worst case'' error model.

\paragraph{Alternative measures of prediction accuracy.}

Consider the following $k$-facility location instance: 

\begin{example}[$2$-facility location sensitivity]
\label{example:non-robustness-of-two-median}
Let $k=2$, and assume the $n$ agents' actual locations $X$ are divided evenly between 0 and 1 on the real line. Let $X'$ be the prediction, with all coordinates predicted accurately except for a single coordinate $i$, which is predicted erroneously to be located at $x'_i=M\gg n$. While the optimal solution for $X$ is to place the $k=2$ facilities at $0$ and $1$ for a total cost of zero, in the optimal solution for $X'$ one of the facilities moves to $M$. 
\end{example}

In this example, a single error causes the optimal solution for prediction $X'$ to perform quite badly for the actual dataset $X$. It thus seems reasonable to deem the prediction error high, and seek a robust design that essentially ignores $X'$. 
However, in many ML contexts, a large prediction error is to be expected --- even excellent predictors will err on a certain fraction of the dataset~(\cite{gupta2022augmenting}).
Moreover, we would expect that if this fraction is not too large, the prediction will still be valuable despite the high ``worst case'' prediction error. In other words, a high prediction error does not always justify ``throwing out'' the prediction and settling for the best worst-case guarantee. 
This motivates alternative measures of prediction accuracy.
Such alternatives seem especially promising for mechanism design, in which in addition to the prediction, we have the reports of the agents to help us disambiguate correct predictions from errors.
A concrete question in this vein is the following: 

\begin{question}[Mechanism design for $2$-facility location on a line]
\label{qu:example}
    In settings such as Example~\ref{example:non-robustness-of-two-median}, can we design a mechanism that 
    uses agent reports to get a good solution for $X$ despite the large worst case prediction error in $X'$? 
\end{question}
In this work we give a positive answer to this question.


\subsection{Our Model} 

We study $\mac$ predictions, where $\mac$ stands for Mostly Approximately Correct. 
The $\mac$ model is a formulation that captures the essential features of some closely related previous models:
e.g., the model of~\citet{azar2022online} introduced in the context of online metric problems,
and the PMAC model of~\citet{balcan2018submodular} 
introduced in the context of learning submodular functions.

A vector of predictions $X'$ is $\mac(\eps, \delta)$ with respect to dataset $X$ if at least a $(1 - \delta)$-fraction (i.e., most) of the predictions are approximately correct up to an additive $\eps$-error. 
Intuitively, if we do not think it is likely that most of our predictions will be close to being correct ---we would typically refrain from relying on them in decision-making processes.
Throughout, $\delta \in [0, 0.5)$ is assumed to be small (but can still be a constant).
Crucially, when a prediction is not approximately correct (i.e., it belongs to the $\delta$-fraction of incorrect predictions), the error is allowed to be arbitrarily large. Also, the $\delta$-fraction of wrong predictions is allowed to be arbitrary among all possible subsets of that size. 
We show that despite their adversarial nature, and the fact that the prediction error is unbounded, such predictions can be used, in combination with strategic reports, to produce a solution with better cost compared to the best no-prediction solution.

The $\mac$ prediction model may be suitable for capturing errors from a variety of sources, mainly:

\begin{itemize}[nosep]

    \item {\it Machine learning}. 
    In practice, ML models produce predictions which are mostly accurate, and the accuracy parameters $\epsilon,\delta$ can be estimated through testing. 
    For example, we may have a predictor with the guarantee that each prediction 
    (independently) will be $\eps$-correct with probability at least $1 - \delta$. 
    Thus, through concentration bounds such as Chernoff, 
    with high probability the number of predictions which are $\eps$-incorrect is at most $\approx \delta n$. This was the motivation in~\cite{balcan2018submodular}.
        
    Note that unlike this example, the $\mac$ model does not require that the predictor's errors are independent.
    
    \item {\it Corruptions}.
    In a setting in which the algorithm's input itself might contain errors, we can treat the corrupted input as a prediction of the true input and apply the $\mac$ model. Corruptions can be a result of adversarial changes to the algorithm's input, e.g., in a malicious attempt to affect the output --- in which case the adversary may change only a $\delta$-fraction of the input to avoid getting caught. Corruptions can also capture \emph{outliers} in the data, e.g., as a result of rare but arbitrary measurement errors. This was the motivation in~\cite{azar2022online}.
    As another example, the input data may have been collected at a certain point in time, while the ground truth continues to evolve 
    --- in the context of facility location, a population census may form the prediction, with errors occurring since a small fraction of the population has relocated since being surveyed. 
    
\end{itemize}

The accuracy parameters $\eps,\delta$ of the $\mac$ model can be viewed as \emph{confidence parameters}. The use of confidence parameters is prevalent in the algorithms with predictions literature (see, e.g.,~\cite{agrawal2022learning}). Arguably, $\eps,\delta$ are highly \emph{explainable} compared to most confidence measures in the literature. 
It is also quite natural to estimate them from data.

\subsection{Our Results}


We study mechanism design with $\mac$ predictions, designing both deterministic and randomized mechanisms for facility location problems.
A preliminary simple observation is that in the context of facility location, it suffices to consider mostly-correct (i.e., $\mac(0,\delta)$) predictions, 
since handling a nonzero approximation parameter $\eps$ follows directly
(see \cref{subsec:the-effect-of-eps}). 
After developing useful tools in Section~\ref{sec:robustness}, 
in Sections~\ref{sec:deterministic}-\ref{sec:rand-two-fac} we design and analyze the following strategyproof anonymous mechanisms with $\mac(0,\delta)$ predictions for the $n$ agent locations: 

\begin{enumerate}
    \item For \emph{single-facility location in $\R^d$}, we design a deterministic mechanism that guarantees an approximation ratio of $\min\{1 + \frac{4\delta}{1-2\delta}, \sqrt{d}\}$ (see \cref{alg:best_choice_single_facility_loc} and \cref{thm:single_fac_loc}).
    For sufficiently small (but still constant) $\delta$, this improves upon the no-prediction guarantee of $\sqrt{d}$ by~\citet{meir2019strategyproof}, which is tight.
    Our mechanism decides, based on the value of $\delta$, whether to use the 
    mechanism of~\citet{meir2019strategyproof} on the location reports, or the $\geomed$ of the predictions as defined and analyzed in Section~\ref{sec:geom-median}.

    \item As our main technical result, for \emph{$2$-facility location on a line}, we design a randomized mechanism that guarantees an expected approximation ratio of $3.6 + O(\delta)$ 
    (see \cref{alg:predict-and-choose-proportionally} and \cref{thm:rand-two-fac}).
    For sufficiently small (but still constant) $\delta$, this improves upon the best-known no-prediction expected guarantee of $4$ by \citet{lu2010asymptotically}.%
    \footnote{The mechanism of~\citet{lu2010asymptotically} works for any metric space.}
    This result provides a positive answer to \cref{qu:example}. We give a brief description of the mechanism in Section~\ref{sub:techniques}. 
    Note that $2$-facility location on a line was studied by~\citet{procaccia2013approximate}, and they provided a \emph{deterministic} mechanism with an $(n-2)$ approximation ratio, which is tight~\cite{lu2010asymptotically,fotakis2014power}. 
    
    
    \item For \emph{$\beta$-balanced $k$-facility location} with constant $k$ (in any metric space), which is a natural extension of facility location 
    where at least a $\beta$-fraction of the $n$ points must be assigned to each of the $k$ facilities (for cost-sharing or load-balancing reasons), 
    we design a simple deterministic strategyproof mechanism. 
    Our mechanism guarantees an approximation ratio of $c + O(\frac{\delta}{\beta})$ where $c$ is a constant (see Algorithm~\ref{alg:balanced_k_facility_loc} and Theorem~\ref{thm:blanaced-k-fac}). For a sufficiently small (but still constant)~$\delta$, this greatly improves upon the best-known no-prediction guarantee of $\nf{n}{2} - 1$ by \citet{aziz2020facility} (we remark that \citet{aziz2020facility} study a variant called \emph{capacitated}, to which our result applies).%
    \footnote{$\beta$-balanced means every cluster is of size at least $\beta n$. This implies a capacity of $n - (k-1)\beta n$ on the cluster size.}

\end{enumerate}

While we have not optimized the constants, our conceptual contribution is in showing that even if the worst case error of the predictions is unbounded, it can still provide useful information for improving the performance of facility location mechanisms when coupled with strategic reports. 
In other words, the $\mac$ model allows achieving robustness to outliers while still beating the no-prediction performance. 
Our single-facility result demonstrates how the accuracy parameter of the model, $\delta$, can serve as a trust parameter which has natural explainability.
Another interesting feature of our results is that utilizing $\mac$ predictions seems to suggest a richer family of mechanisms than standard interpolations of no-prediction and complete trust of the predictions.

\subsection{Our Techniques}
\label{sub:techniques}

Our main result requires a combination of techniques as we now describe.

\paragraph{Robust statistics.} 
Of possible independent interest, we develop quantitative versions of classic results in robust statistics on \emph{breakdown points}.
Consider a \emph{location estimator}, a function that gets $n$ points in a metric space and returns $k$ locations. 
Its breakdown point is the proportion $\delta$ of adversarially-modified points it can handle before returning an arbitrarily far result from the original estimator. It is well-known that the mean has a breakdown point of $\delta=\nf{1}{n}$ (consider changing one point to an arbitrarily large value)
and that the median is robust with a breakdown point of $\delta=\nf{1}{2}$. 
However, the notion of breakdown point is qualitative, and does not reveal the relation between the fraction $\delta$ of modified points, and how much the estimator changes. 
In \cref{sub:notions} we define two measures of $\delta$-robustness applying to location estimators: distance robustness (\cref{def:dist-robust}) and approximation robustness (\cref{def:approx_robust}). 
Distance robustness upper-bounds the distance between the original and modified location estimator as a function of $\delta$, while approximation robustness is with respect to a cost function (e.g., the distances sum of the points from the estimator), and upper-bounds the change in cost assigned to the location estimator. We also show a connection between the two notions (Lemma~\ref{lem:switch}).

\paragraph{The robustness of the median and its generalizations.} 

In Section~\ref{sec:geom-median} we consider the robustness of $1$-median and $k$-medians when applied to a dataset $X\in V^n$ in a metric space $(V,d)$ (where $k$-medians is $k$ locations that minimize the total distance of all points in $X$ to their nearest center among these $k$ locations).
For $1$-median (also known as the geometric median) we show good distance robustness of $O_\delta(1) \cdot \MAD(X)$ where $O_\delta(1)$ is a small $\delta$-dependent constant, and $MAD$ stands for Mean Average Deviation, a known notion of variance defined as $\MAD(X) := \frac{1}{n}\sum_{x_i \in X}{d(x_i, \geomed(X))}$. An interpretation of our result is that given a small $\delta$, the robustness of the geometric median depends on the $\MAD$ of the points,%
\footnote{Intuitively, for data with high $\MAD$ (or variance), there is no single point that represents it in a good way.}
and the quantification of this relation is given by \cref{thm:one-med-dist-robust} and \cref{1_median_approx_robustness_results}. There is some notion of robustness in the robust statistics literature for the geometric median: Given a dataset in $\R^d$ generated by a Gaussian distribution with $\delta$ fraction of outliers, the distance between the mean of the distribution and the geometric median of the dataset is $\Omega(\delta \sqrt{d})$ ~(\cite{lai2016agnostic}). However, we are not aware of other known results such as ours.
The robustness of $1$-median is used for mechanism design with $\mac$ predictions for $1$-facility location, as well as for our main result.

In \cref{sub:k-medians} we show that for $k > 1$ there is no hope for robustness of $k$-medians. We thus turn to $\beta$-balanced $k$-medians, which are $k$ medians that induce clusters of size at least $\beta n$, where $n$ is the total number of points (see \cref{def:center-induced}). Clustering with the additional constraint of minimum cluster size has applications such as data privacy (\cite{aggarwal2010achieving}). The robustness of $\beta$-balanced $k$-medians implies that balanced clustering algorithms can be robust to adversarial corruptions or outliers. 


\paragraph{The second facility location problem.} 
In \cref{subsec:second-facility-loc-problem} we introduce the second facility location problem. 
It is often the case that facilities are created incrementally and not all at once. 
Hence, it is natural to ask how can we design a mechanism for choosing the location of a new facility (say a second hospital) given an existing one. We show that a simple variant of the \emph{proportional mechanism} by \citet{lu2010asymptotically} achieves a (tight) approximation ratio of $3$ for the line metric. 


\paragraph{Putting it all together: Beyond interpolation.}

We use a combination of the above technical results to prove that our randomized mechanism for $2$-facility location on a line has good properties. In Section~\ref{sec:big-cluster-center-robustness} we define a new location estimator named $\bcc$ and show its $\delta$-robustness. We also use the $\delta$-robustness of $\beta$-balanced $k$-medians and the mechanism for the second facility location problem on a line. Our mechanism differs from many existing designs which cleverly interpolate between the two approaches of trusting the predictions and using the no-predictions algorithm.


\subsection{Related Work}

\paragraph{Facility location mechanism design with predictions.}
In the ``worst case'' error model, \citet{agrawal2022learning} show a mechanism for \emph{single-facility location} on the real plane $\R^2$.
Their prediction consists of a single value representing advice for where to locate the facility.
Given a trust parameter $c \in [0,1]$, their mechanism guarantees an approximation ratio of at most $\frac{\sqrt{2c^2 + 2}}{1+c}$ if the prediction is perfect, and $\frac{\sqrt{2c^2 + 2}}{1-c}$ in the worst case (if the prediction is arbitrarily bad).
In theory, we could reduce the $\mac$ model to the one studied by \citet{agrawal2022learning}, by setting their single prediction to be the geometric median of the $\mac$ predictions. Then $c$ can be determined according to $\delta$. However, knowing $\delta$ and the $\delta$-robustness of $\OneMed$, we can directly choose either the no-predictions mechanism or the predictions' $\OneMed$, rather than interpolating between these two solutions.

Mechanism design for the \emph{$2$-facility location} on the line with predictions is studied by \citet{ijcai2022p81}, where the predictions are for the locations of the agents. While they do not define an error parameter, their deterministic mechanism is in line with the ``worst case'' prediction error model: it achieves an approximation ratio of $\frac{n}{2}$ if the predictions are perfect, and an approximation ratio of $2n$ if they  are arbitrarily bad. 
\citet{DBLP:journals/corr/abs-2212-09521} study mechanism design with predictions for obnoxious facility location where agents prefer to be as far as possible from the facilities.

\paragraph{Two-facility location on a line mechanism design without predictions.}

In the setting of no-predictions,  the best strategyproof deterministic anonymous mechanism achieves a tight approximation ratio of $n-2$~(\cite{procaccia2013approximate, fotakis2014power}). In comparison, in our $\mac$ model under the additional assumption that the number of agents assigned to each facility in the optimal solution is of minimum size $\Omega(\delta n)$, the deterministic mechanism 
 in \cref{alg:balanced_k_facility_loc} achieves a constant approximation ratio. Under no further assumption, our randomized mechanism in \cref{alg:predict-and-choose-proportionally} achieves a constant expected approximation ratio for small $\delta$.

\paragraph{Models of prediction accuracy.}
\citet{azar2022online} study online graph algorithms with a notion of error they call \emph{metric error with outliers}, applied to problems where the predictions are locations in a metric space. 
The error parameters of this model are $D$ and $\Delta$---roughly, $D$ is the minimum cost matching between the predictions and the actual locations, where $\Delta$ outliers may be excluded from the matching. These parameters are similar to the $(\eps, \delta)$ parameters of the $\mac$ model. They then give a general framework for solving online graph problems in their model.
Interestingly, their matching cost $D$ captures a notion of average error, whereas we focus on individual error; while our results can be extended to fit within their model, it would be interesting to investigate a version of $\mac$ where $\eps$ captures some notion of average error.

\citet{GkatzelisKST22} generalize the outlier model of~\citet{azar2022online} and apply it to improve the price of anarchy of cost-sharing mechanisms. \citet{xu2022learning} study online Steiner tree problems with predictions where the error parameter is the number of erroneous predictions. \citet{bernardini2022universal} introduce a model named the ``cover error model'' that generalizes the model of~\citet{azar2022online} to achieve more precision for online problems.

\citet{gupta2022augmenting} study predictions that are incorrect independently with probability $\delta$.
Their work focuses on online algorithms rather than mechanism design. Our model is related to their model since independent errors imply $\mac$ predictions w.h.p.~by Chernoff. The $\mac$ model is more general in the sense that it does not assume independence of erroneous predictions. \citet{emek2023online} also consider a similar model to the one of \citet{gupta2022augmenting}, assuming access to predictions that are correct with some probability or contain random bits otherwise. \citet{dutting2021secretaries} consider a signaling scheme advice model for the secretary, which for binary ML-advice roughly translates into measures of accuracy of the signals (named recall and specificity).

\citet{jiang2021onlineselection} consider a prediction guaranteed to be inside a prediction interval. While their error model is closer to the ``worst case'' error model, their motivation is confidence intervals from statistics, where it is reasonable to assume that a confidence interval is achieved with probability $1 - \delta$ for some $\delta \in [0,1]$. \citet{sun2023online} consider an approach for ``uncertainty quantified'' prediction. One class of predictions they focus on, the ``probabilistic interval prediction'', can also be seen as a a confidence interval, as in the motivation of \citet{jiang2021onlineselection}.
Like our $\mac$ model, they ``model in'' the amount of trust in the advice. \citet{sun2023online} focus on settings where the algorithm receives a single prediction (such as the ski rental problem). They use online learning for tuning the trust parameters, and introduce a regret minimization analysis for learning the trust parameter online (over many instances of the same problem). The $\mac$ model addresses problems that necessitate the quantification of uncertainty for multiple predictions simultaneously within the same input instance, such as in the $k$-facility location problem.

Additional related work is deferred to Appendix~\ref{sub:more-related}.

\section{Preliminaries}
\label{sec:notation}

Let $(V, d)$ be a metric space and $n \in \N$. 

\subsection{Predictions}

The input to the facility location problem is a
multi-set $X = \{x_1, \ldots, x_n\} \in V^n$. 
We are also given a
\emph{prediction}, which is another multi-set
$X' = \{x'_1, \ldots, x'_n\} \in V^n$. 
The distance between a point $u$ and a multi-set $W$ is $d(u,W)=\min_{w\in W} d(u,w).$
The \emph{Hausdorff distance}
between multi-sets $U,W \in V^n$ is
$$
d_H(U,W) := \max\{ \max_{u \in U} d(u,W), \max_{w \in W} d(w,U)\},
$$ 
that is, the maximum distance between a point in one multi-set and the other multi-set.



\begin{definition}
  Fix $\eps > 0$. For $x,x'\in V$, we say $x'$ is an
  \emph{$\eps$-correct} prediction of $x$ if $d(x, x') \le \eps$,
  otherwise we say it is \emph{$\eps$-incorrect}.
\end{definition}
\begin{definition}
  The \emph{$(\eps,r)$-neighborhood} of multi-set $X \in V^n$ is defined to be all
  predictions $X'\in V^n$ such that
  $\abs{\{i \in [n] \mid x'_i\text{ is an }\eps\text{-incorrect prediction
    of }x_i\}} \leq r$. 
\end{definition}

The $(0,r)$-neighborhood of $X$ is often just
  referred to as the \emph{$r$-neighborhood} of $X$.



\begin{definition}[Mostly Approximately Correct (\mac) Predictions] \label{def:mac-pred}
  Fix $\eps > 0$, $\delta \in [0, 0.5)$.
  The point set $X'$ is an \emph{$(\eps, \delta)$-\mac prediction} for $X$ if $X'$
  belongs to the $(\eps, \delta|X|)$-neighborhood of $X$.
\end{definition}


\subsection{Location Estimators}

\begin{definition}[k-medians cost function]\label{def:kmed-cost}
Let $k \in \N$. Given multi-sets  $X \in V^n, F \in V^k$, the \emph{$k$-medians cost function} is
\begin{gather}
  \kmed(X,F) := \sum_{x \in X} d(x,F).
\end{gather}
In the context of the facility location problem, this function is also known as the social cost function or the utilitarian goal function.
\end{definition}

The \emph{$k$-medians problem} takes as input a multi-set $X \in V^n$
and outputs a multi-set $F$ that minimizes this cost function. The minimizer $F^* \in V^k$ the problem is called the \emph{$\KMedSol$ solution}. Solving the $k$-medians problem can be viewed as finding a \emph{location estimator} of $X$ that optimizes the $k$-medians objective function:

\begin{definition}[Location Estimator]\label{def:loc-estimator}
    For $n, k \in \N$, a function $f(X): V^n \rightarrow V^k$ is called a location estimator.
\end{definition}

Some common examples of location estimators with $k=1$ for points on the real
line are the minimum, maximum,
mean, and median; for general metric spaces, $k$-means and $k$-medians are well-studied examples ($k$-means is similar to $k$-medians with squared distances).

For $k=1$ and $V = \R^d$, the $k$-medians solution is also called the $\geomed$ which is a generalization of the median to higher dimensions. A different such generalization of the median is the $\cwmed$:

\begin{definition}
    For a multi-set of $n$ points $X \subseteq \R^d$: $\cwmed(X) := (l_1, \ldots, l_d)$ where where $l_j$ is the median of the multi-set of the $j$'th coordinates of $X$ for all $j \in [d]$.
\end{definition}

\eat{
  
For the case
of $k = 1$, the quantity $\frac1{|X|} \kmed[1](X,\{v\})$ is often called
the \emph{mean absolute deviation} of $X$ around $v$, and denoted
$\mad(X,v)$.
Now the \emph{mean absolute deviation} of a point set $X$ is defined
as $\mad(X) := \min_{v \in V} \mad(X,v)$. Finally, the \emph{$\OneMed$} or \emph{geometric median} of the multi-set $X$, denoted by $\med(X)$, is defined to be a point $m \in V$ that achieves the minimum above.

\ag{Check if this definition is widely used, or only in one section.}
\zohar{I think it is only used in one section. I originally just defined it there (in the 1-median robustness section}.

\begin{definition}[Induced clustering]
  A collection of centers $F = \{f_1, \ldots, f_k\} \sse V^k$
  \emph{induces a partition} $\calC = \{C_1, \ldots, C_k\}$ of the
  dataset $X$ if
  $C_i \sse \{x \in X \mid d(x,f_i) \leq d(x,f_j) \, \forall
  j\}$. This partitions is called the clustering of $X$ induced by $F$. 
\end{definition}
Note that a collection $\calC$ need not be unique, since points in
$x$ can choose between any of their closest centers in $F$.

} 









\section{\texorpdfstring{$\delta$}{}-Robustness of Location Estimators}
\label{sec:robustness}

In this section we define and relate two notions of robustness to change in a $\delta$-fraction of the dataset (Section~\ref{sub:notions}). We analyze the $\delta$-robustness of 1-median and $k$-medians in Sections~\ref{sec:geom-median} and~\ref{sub:k-medians}, respectively. 
For $k$-medians we must add a \emph{balancedness} condition on the clusters to obtain robustness, and show its necessity.
In later sections we use these robustness results for mechanism design with $\mac$ predictions.
Note that our focus is on ``mostly correct'' datasets; we return to ``approximately correct'' as measured by the parameter $\eps$ in later sections.

\subsection{Robustness Notions}
\label{sub:notions}

We consider two robustness notions for location estimators, where the robustness is with respect to changing a $\delta$-fraction of the points (hence referred to as $\delta$-robustness):
The first notion 
measures the movement in the solution (as measured by the 
Hausdorff distance), and the second measures the change in an objective function applied to the solution. In Lemma~\ref{lem:switch} we relate the two notions.

\begin{definition}[Distance Robustness]
\label{def:dist-robust}
  Let $\rho\ge 0$, $\delta \in [0,0.5)$.
  For location estimators $f, \widehat{f}: V^n \rightarrow
  V^k$, 
  we say that
  $\widehat{f}$ is $(\rho,\delta)$-\emph{distance-robust} with respect to
  $f$ if for any $X \in V^n$ and any $X'$ in the $\delta \abs{X}$-neighborhood of $X$,
  \[
  d_H\prn*{f(X), \widehat{f}(X')} \le \rho.
  \]
  If  $\widehat{f} = f$, we say that $f$ is $(\rho,\delta)$-distance-robust.
\end{definition}

In the next definition, cost function $F$ is evaluated
for the same dataset $X \in V^n$ with respect to two different solutions $f(X)$ and
$\widehat{f}(X')$. 

\begin{definition}[Approximation Robustness]\label{def:approx_robust}
  Let $\gamma\ge 1$, $\delta \in [0,0.5)$
  and let $F: V^n \times V^k \to \R_{\geq 0}$ be a
  cost measure.  For location estimators $f, \widehat{f}: V^n \rightarrow V^k$, we
  say that $\widehat{f}$ is a
  $(\gamma,\delta)$-\emph{approximation-robust solution} for cost
  function $F$ with respect to $f$ if for any $X \in V^n$ and any $X'$
  in the $\delta \abs{X}$-neighborhood of $X$,
  \begin{equation}
  F(X, \widehat{f}(X')) \leq \gamma \cdot F(X, f(X)).\label{eq:approx}
  \end{equation}
  If $\widehat{f} = f$, we say that $f$ is a 
  $(\gamma,\delta)$-approximation-robust solution for $F$.
  If Eq.~\eqref{eq:approx} only holds for datasets $X \in \calY$
  for some $\calY \subseteq V^n$, we say that $f$ is a $(\gamma,
  \delta)$-approximation-robust solution for $F$ restricted to instances $\calY$.
\end{definition}

The following lemma shows
that distance robustness 
implies approximation robustness
if $F$ is the $k$-medians cost function.

\begin{lem}
  \label{lem:switch}
  Consider location estimators $f, \widehat{f}: V^n \to V^k$ that satisfy the following two properties:
  \begin{enumerate}[nosep]
  \item For any $X \in V^n$ and $X'$ in the $\delta \abs{X}$-neighborhood of $X$, 
    \[ \kmed(X', \widehat{f}(X')) \leq \kmed(X', f(X));\] 
  \item $\widehat{f}$ is $(\rho,\delta)$-distance-robust with respect to $f$.
  \end{enumerate}
  Then $\widehat{f}$ is a $\prn*{1 + \frac{2\delta \abs{X} \rho}{\kmed(X, f(X))},
  \delta}$-approximation-robust solution for $F = \kmed$ with respect
  to $f$.
\end{lem}

\begin{proof}
  Let $n = \abs{X}$, let $A := \{i \in [n] \mid x_i = x'_i\}$ be the indices where $X,X'$
  are the same, and $B := [n] \setminus A$ be the remaining indices. Define
  $G := f(X)$ and $H := \widehat{f}(X')$. Using this notation we can write $\kmed(X,\widehat{f}(X'))$ as 
  \begin{align}
    \sum_i d(x_i, H) &= \sum_i d(x'_i, H) +  \sum_{i \in B}
                       (d(x_i, H) - d(x_i', H)) \notag \\
                     &\leq \sum_i d(x'_i, G) +  \sum_{i \in B}
                       (d(x_i, H) - d(x_i', H))  \label{eq:1} \\
                     &=  \sum_i d(x_i, G) +  \sum_{i \in B}
                       (d(x_i', G) - d(x_i, G))  +  \sum_{i \in B}
                       (d(x_i, H) - d(x_i', H))  \notag \\
                     &=  \kmed(X,f(X)) +  \sum_{i \in B}
                       (d(x_i', G) - d(x_i', H))  +  \sum_{i \in B}
                       (d(x_i, H) - d(x_i, G)), \label{eq:differences}
  \end{align}
  where (\ref{eq:1}) uses the first assumption of the theorem.

  We claim that each of the differences in Eq.~\eqref{eq:differences} 
  is at most
  $d_H(G,H)$. Indeed, for any $x \in V$, let its closest points in
  $G$ and $H$ be $g$ and $h$ respectively. Then
  $d(x,g) - d(x,h) \leq d(g,h)$ by the triangle inequality, and this
  is at most the Hausdorff distance; the case of $d(x,h) - d(x,g)$ is
  identical. Therefore, 
  $$\kmed(X,\widehat{f}(X')) \leq \kmed(X,f(X)) +
  2\abs{B} \cdot d_H(G,H).$$ 
  Finally, using the definition of
  $(\rho,\delta)$-distance-robustness implies that $d_H(G,H)$ is at
  most $\rho$. The fact that $|B| \le \delta |X|$ completes the proof. 
\end{proof}

\subsection{The \texorpdfstring{$\delta$}{}-Robustness of \texorpdfstring{$1$}{}-Median}
\label{sec:geom-median}

In this section, we quantify the distance robustness (and hence approximation robustness, by Lemma~\ref{lem:switch}) of the 1-median location estimator. We show that changing any $\delta$-fraction of a dataset $X$ where $\delta < \nf12$ 
can move the $\OneMed$ by only $\frac{2}{1-2\delta}$ times the average cost of a point in the optimal solution of $X$.
Hence, this changes the total $\OneMed$ cost by only a $1 + O(\frac{\delta}{1-2\delta})$ factor.


\begin{thm}
  \label{thm:one-med-dist-robust}
    For $\delta < \nf12$, the $\OneMed$ location estimator is $(\rho,\delta)$-distance-robust, where $\rho = \frac{2}{(1-2\delta)}\cdot \frac{\kmed[1](X)}{\abs{X}}$. 
\end{thm}

\begin{proof}
Let $X  = \{x_1, \ldots, x_n\}, X' = \{ x'_1, \ldots, x'_n \} \in V^n$ where $X'$ is in the $\delta \abs{X}$-neighborhood of $X$.
  Let $n = \abs{X}$, let $A := \{i \in [n] \mid x_i = x'_i\}$ be the
  indices where $X,X'$ are the same, and $B := [n] \setminus A$ be the
  remaining indices. Let $m := \arg\min_{v \in V} \kmed[1](X,\{v\})$
  and $m' := \arg\min_{v \in V} \kmed[1](X',\{v\})$. It follows that
  $\sum_{i\in[n]} d(x_i',m') \leq \sum_{i\in[n]} d(x_i', m)$.  Using
  the triangle inequality and the fact that $x_i' = x_i$ for all
  $i \in A$, we get that
  \begin{align*}
    \sum_{i \in A} (d(m',m) - d(m,x_i)) + \sum_{i \in B} (d(x_i',m) -
    d(m',m)) \leq \sum_{i \in A} d(x_i,m) + \sum_{i \in B} d(x'_i,m)
    \\ \implies d(m',m)\cdot (\abs{A} - \abs{B}) \leq 2 \sum_{i \in A}
    d(x_i,m) \leq 2 \sum_{i \in [n]} d(x_i,m) = 2\, \kmed[1](X).
  \end{align*}
  Since $\abs{B} \leq \delta n$ and $\abs{A}+\abs{B} = n$, we get $d(m',m) \leq
  \frac{2}{(1-2\delta)n} \kmed[1](X)$.
\end{proof}

The quantity $\frac{\kmed[1](X)}{|X|}$ is called the \emph{mean
absolute deviation} and can be viewed as a notion of variance.
Hence, Theorem~\ref{thm:one-med-dist-robust} essentially says that since $\delta < \nf12$, 
the $\OneMed$ only
 changes by $O_\delta(1)$ (a constant depending on $\delta$)
times the mean absolute deviation.


Substituting the distance-robustness parameter of $\rho = \frac{2}{1-2\delta} \frac{\med(X)}{|X|}$ into  \Cref{lem:switch} gives us the following result for approximation
robustness:

\begin{cor} 
\label{1_median_approx_robustness_results} 
For $\delta < \nf12$, the $\OneMed$ location estimator is $(1 +
\frac{4\delta}{1-2\delta},\delta)$-approximation-robust for
$\kmed[1]$. 
\end{cor}

Corollary~\ref{1_median_approx_robustness_results} says that if $X'$ is obtained by perturbing a $\delta$-fraction
of the points in $X$, then the geometric median of $X'$ is a $(1 + \frac{4\delta}{1-2\delta})$-approximately optimal solution for the original point set $X$ (with respect to the $\OneMed$ cost objective). 
This result can be viewed as a quantitative version of a classical result of \citet{lopuhaa1991breakdown}, which shows that the
\emph{breakdown point} for the geometric median is $\nf12$. 
This gives us an understanding not only of the fraction of perturbations required to ``break'' the estimator (i.e., make it arbitrarily far from the real location), but also of the deterioration in the estimation before breaking it.
The quality of estimation as a function of $\delta$ is captured by our result of $\rho = \frac{2}{1-2\delta} \frac{\med(X)}{|X|}$, showing that a smaller $\delta$ corresponds to better estimation.

We conclude the discussion of $\delta$-robustness of the
  $\OneMed$ with two remarks:
\begin{enumerate}[nosep]
\item Results similar to~\Cref{thm:one-med-dist-robust} hold even for the case where $X'$ is
  obtained by adding or removing points; in fact, the approximation
  factors can be improved slightly (the proof is almost the same). 
\item The quantitative bounds are tight; consider the example where
  $X$ comprises of $(0.5 - \delta)n + 1$ points at $0$, and
  $(0.5 + \delta) n - 1$ points at $1$ 
  and $X'$ is obtained from $X$
  by moving 
  $\delta n$ points from $1$ to $0$. A calculation shows
  that in this case the approximation robustness of the geometric
  median is $(1 + \frac{4\delta}{1-2\delta} - \Theta(\frac{1}{n}), \delta)$. 
\end{enumerate}

\subsection{The \texorpdfstring{$\delta$}{}-Robustness of \texorpdfstring{$\beta$}{}-Balanced \texorpdfstring{$k$}{}-Medians}
\label{sub:k-medians}

Given the robustness of the $\OneMed$, it is natural to ask if the same holds for
$k$-medians. Unfortunately, even for $2$-medians, changing a single point can
cause the Hausdorff distance between the solutions to be arbitrarily
large. 

Recall \cref{example:non-robustness-of-two-median} from Section~\ref{sec:intro}.
  In this example, there is no choice of $k=2$ centers that simultaneously achieves a good (bounded) approximation for the $2$-medians problem on both $X$ and $X'$. 
Given this negative example, we turn to a \emph{balanced} version of $k$-medians
for which we are able to give a robustness guarantee. In \S\ref{subsec:balanced-k-fac} we will use this for a deterministic strategyproof for balanced-$k$-facility location in a general metric space. In \S\ref{sec:rand-two-fac} we will
use this to give our random strategyproof mechanism for $2$-facility location on a line.

\begin{definition}[Center-Induced Partitions]
  \label{def:center-induced}
  A collection of centers $F = \{f_1, \ldots, f_k\} \sse V^k$
  \emph{induces a partition} $\calC = \{C_1, \ldots, C_k\}$ of the
  dataset $X$ if
  $C_i \sse \{x \in X \mid d(x,f_i) \leq d(x,f_j) \, \forall
  j\}$. This partition is called the clustering of $X$ induced by $F$. This partition (or clustering) is \emph{$\beta$-balanced} if
  $|C_i| \geq \beta |X|$ for all $i \in [k]$.
\end{definition}

Note that a collection $\calC$ need not be unique, since points in
$X$ can ``choose'' between any of their closest centers in $F$.

\begin{definition}[$\beta$-balanced $k$-median]
  \label{def:balanced-kmed}
  The \emph{$\beta$-balanced $k$-medians problem} takes a dataset
  $X \in V^n$ as input, and outputs centers $F \in V^k$ along with an
  induced $\beta$-balanced partition $\calC$ of $X$; the goal is to
  minimize the cost function $\kmed(X,F) = \sum_{x \in X} d(x,F)$.
\end{definition}

This problem is similar to the \emph{lower-bounded
  $k$-median} problem where we want to output
$F = \{f_1, \ldots, f_k\} \in V^k$ as well as a partition
$\{C_1, \ldots, C_k\}$ of $X$ such that $|C_i| \geq \beta |X|$, in
order to minimize $\sum_{x \in C_i} d(x,f_i)$ (see \cite{han2020approximation}, \cite{wu2022approximation}). A significant
difference between that definition and ours is that we require the
partition to be induced by $F$ (in the sense of
\Cref{def:center-induced}), which is not required by lower-bounded
$k$-median.

The following theorem shows that if we consider balanced $k$-median,
then computing the best ``slightly less balanced'' solution on the
predictions $X'$ has a $k$-medians cost (on the original dataset $X$)
that is close to the optimal balanced $k$-medians cost for $X$.

\begin{thm} \label{thm_balanced_k_median_robustness} Let $b >
  2k+2$. Consider the algorithm $\calB$ that computes the optimal
  $(b-1)\delta$-balanced $k$-medians solution on its input.  For any
  instance $X \in V^n$ of the $b\delta$-balanced $k$-medians problem
  with optimal solution $G$, let $X' \in V^n$ belong to the
  $\delta$-neighborhood of $X$. The algorithm $\calB$, when given
  $X'$, returns a solution $H$ such that
  \[ d_H(G,H) \leq \frac{2k}{\delta|X|\cdot (b-2-2k)} \cdot \sum_{x
      \in X} d(x,G). \] Moreover,
the $k$-medians cost  $\kmed(X,H) \leq (1 + \frac{4k}{b-2-2k}) \kmed(X,G)$, and $H$ 
  induces a $(b-2)\delta$-balanced partition $\calC_H$ of $X$.
\end{thm}

\begin{proof}
  The claimed algorithm $\calB$ is simple: compute the optimal
  $(b-1)\delta$-balanced $k$-medians solution on $X'$. Since $X$ admits
  a $b\delta$-balanced solution $G = \{g_1, \ldots, g_k\}$, and $X'$
  differs from $X$ in only $\delta|X|$ points, this solution $G$ is
  $(b-1)\delta$-balanced for $X'$. Now let $H = \{h_1, \ldots, h_k\}$
  be a $(b-1)\delta$-balanced $k$-medians solution of least cost for
  $X'$ computed by $\calB$: this means
  \[ \sum_{x' \in X'} d(x',H) \leq \sum_{x' \in X'} d(x',G). \] Let
  $\calC_H' = \{C_1', \ldots, C_k'\}$ be the balanced partition of
  $X'$ induced by $H$.  Similarly, since $X$ is in the
  $\delta$-neighborhood of $X'$, the centers in $H$
  induce a $(b-2)\delta$-balanced partition of the original dataset
  $X$; call this $\calC_H = \{C_1, \ldots, C_k\}$. Finally, let
  $\calC_G = \{C_1^*, \ldots, C_k^*\}$ be the $b\delta$-balanced
  partition of $X$ induced by $G$. Recall that $n = |X|$; define 
    \begin{align*}
        OPT &:= \kmed(X,G) = \sum_{x \in X} d(x, G) = \sum_{i = 1}^k \sum_{x \in C_i^*}
        d(x,g_i) \quad \text{and}\\
        ALG &:= \kmed(X,H) = \sum_{x \in X} d(x, H) = \sum_{j = 1}^k \sum_{x \in C_j}
    d(x,h_j). 
    \end{align*}

  We now want to prove that $\max_i d(g_i, H)$ and $\max_j d(h_j, G)$
  are both small. To show the former, fix any $i$ and let $C_i^*$ be
  the cluster corresponding to $g_i$. Since the clustering is
  $b\delta$-balanced, $|C_i^*| \geq b\delta n$. By averaging, there
  exists some cluster $C_j \in \calC_H$ such that
  $|C_i^* \cap C_j| \geq b\delta n/k$; choose this $j$, and consider
  the corresponding center $h_j \in H$. Now
  \[ d(g_i, H) \leq d(g_i, h_j) \stackrel{(\star)}{\leq}
    \frac{1}{|C_i^* \cap C_j|} \sum_{x \in C_i^* \cap C_j} \big[
    d(x,g_i) + d(x,h_j) \big] \leq \frac{ALG + OPT}{b \delta n/k}, \]
  where $(\star)$ uses the triangle inequality.  A similar argument
  shows:
  \[ d(h_j, G) \leq \frac{ALG + OPT}{(b-2) \delta
      n/k}. \]

  Hence the Hausdorff distance between the two solutions is $d_H(G,H) \leq \frac{ALG + OPT}{(b-2) \delta
    n/k}$. Using \Cref{lem:switch} with this bound on the Hausdorff distance, we get
  \[ ALG \leq OPT + 2\delta n \cdot d_H(G,H) \leq OPT + 2k
    \cdot \frac{ ALG + OPT }{(b-2)}. \]
  Simplifying, we get $ALG \leq (1 + \frac{4k}{b-2-2k})OPT$, and that
  $d_H(G,H) \leq \frac{2k}{b-2-2k}\cdot \frac{OPT}{\delta n}$, as claimed.
\end{proof}

For a fixed value of $k$, this theorem gives a solution $H$ that is
$(b-2)\delta$-balanced and also 
$(1+ O(\nf1b), \delta)$-approximation-robust with respect to the best
$b\delta$-balanced solution.

\section{Deterministic Mechanism Design for Facility Location}
\label{sec:deterministic}

In this section we show how to design deterministic mechanisms for facility location problems that utilize the $\mac$ predictions to get better approximation guarantees. We do so by utilizing the $\delta$-robustness results from Section~\ref{sec:robustness}. In the mechanism design setting, there are $n$ agents with (true) locations  $X = \{ x_i \mid i \in [n]\} \in \RR^d$, where $d$ is the dimension of the facility location problem. The mechanism has access to location reports by the strategic agents, and also to $X' = \{ x'_1, \ldots, x'_n \}$, which are the $\mac(\eps,\delta)$ predictions of the true locations $X$. 

\begin{figure}[h]
    \centering
    \includegraphics[scale=0.25]{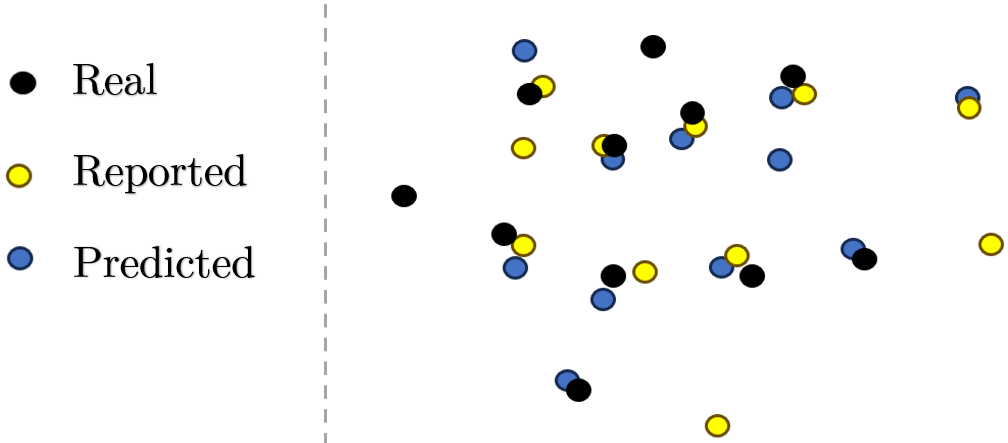}
    \caption{Illustration of the problem settings. The black points are the real agent locations. The yellow points are the locations reported by the agents and the blue points are the predicted locations. For a strategyproof mechanism, the yellow points and the black points will overlap.}\label{facility_loc_with_pred_illustration_figure}
\end{figure}

\subsection{Single Facility Location in \texorpdfstring{$\mathbb{R}^d$}{}}
\label{sec:single-fac}

We start with the single facility location problem ($k=1$). We first assume that $\eps=0$, and later explain the effect of reintroducing $\eps$. In this setting we assume $\delta$ (or an upper bound value) is known and thus we can simply check in advance if $1 + \frac{4\delta}{1-2\delta} \le \sqrt{d}$ or not, and use the algorithm that guarantees us the best approximation ratio.

\begin{algorithm}[H]
	\SetAlgoNoLine
    \caption{Best-Choice-Single-Facility-Loc}
    \label{alg:best_choice_single_facility_loc}
    \KwIn{$\delta\in [0,\frac{1}{2})$ and $X, X' \subseteq \R^d$ where $X$ is reported by strategic agents and $X'$ contains the $\mac$ predictions for $X$}
	\KwOut{A facility location in $\R^d$}
    \leIf{($1 + \frac{4\delta}{1-2\delta} \le \sqrt{d}$)\\}
    {return $\geomed(X')$\\}{return $\cwmed(X)$}
\end{algorithm}

\begin{thm} \label{thm:single_fac_loc}
\cref{alg:best_choice_single_facility_loc} is strategyproof and gets an approximation ratio of $min\prn*{1 + \frac{4\delta}{1-2\delta}, \sqrt{d}}$.
\end{thm}
\begin{proof}
  By \cref{1_median_approx_robustness_results}, computing the
  $\OneMed$ on $X'$ results in a
  $1 + \frac{4\delta}{1 - 2\delta}$-approximation ratio: this is
  because the metric space is $\mathbb{R}^d$, and hence the $\OneMed$
  is in fact the geometric median (which is also known as the spatial
  median, $L_1$ median or Fermat point). Since the geometric median
  computation depends solely on the predictions it is strategyproof. 
  Thus, by the the strategyproofness of $\cwmed$, and by the fact that the decision of which mechanism to use does not depend on the reported locations, \cref{alg:best_choice_single_facility_loc} is strategyproof. Due to the $\sqrt{d}$
  approximation ratio of $\cwmed$ (\cite{meir2019strategyproof}), we get a
  $\min\prn*{1 + \frac{4\delta}{1-2\delta}, \sqrt{d}}$ approximation
  ratio for \cref{alg:best_choice_single_facility_loc}.  
\end{proof}

For general dimensions $d$ and number of points $n$, there is no known formula or algorithm to find the geometric median exactly. However, since the problem is a convex optimization problem, there are optimization algorithms that find solutions with arbitrary precision efficiently. By using such an optimization algorithm we get an additional factor of $\eps' > 0$ where $\eps'$ is arbitrarily small. For more information on efficiently computing the geometric median, see \cite{eckhardt1980weber}, \cite{beck2015weiszfeld} \cite{cohen2016geometric}.

\subsubsection{A combined strategy for the case of high probability
  \texorpdfstring{$\mac(\eps,\delta)$}{}}

The $\mac(\eps,\delta)$ model requires that at most a
  $\delta$ fraction of the points have error more than $\eps$; one can
  extend this to the setting where this requirement holds only with
  high probability. For instance,
consider the setting where the number of prediction errors can be
larger than a $\delta$ fraction with probability at most
$o(\nf{1}{n})$. In this case 
the computed geometric median might be arbitrarily far away with this
small probability, and hence the expected  approximation ratio for our
mechanism would be unbounded. 
To avoid this problem, we can use the agent reports. Intuitively the agents have the incentive to make sure that the returned facility location is not infinitely far away from them.

One way to do this is by using the $\minbb$ mechanism, introduced by
\cite[Mechanism 2]{agrawal2022learning} for the egalitarian cost
function for $d=2$; this naturally extends to any $d \ge 1$. The
mechanism calculates the minimum bounding box $B$ containing all the
input points; then, given a single prediction $o$, it returns the point
$\hat{o} \in B$ closest to $o$. We show in \cref{sec:min-bb} that the
approximation ratio of the generalized $\minbb$ mechanism for the
utilitarian cost function ($\kmed[1]$) is $O(n)$. 

Therefore, by composing $\minbb$ with
\cref{alg:best_choice_single_facility_loc} we get a strategyproof
mechanism which has an approximation ratio of $min\prn*{1 +
  \frac{4\delta}{1-2\delta}, \sqrt{d}}$ with high probability, and has
an $O(n)$ approximation ratio with probability $o(\nf{1}{n})$; this gives an expected approximation ratio of at most $min\prn*{1 + \frac{4\delta}{1-2\delta}, \sqrt{d}} + o(1)$.
The form of the resulting mechanism is simple:

\begin{algorithm}[H]
	\SetAlgoNoLine
    \caption{Bounded-Best-Choice-Single-Facility-Loc}
    \label{alg:bounded_best_choice_single_facility_loc}
    \KwIn{$\delta\in [0,\frac{1}{2})$ and $X, X' \subseteq \R^d$ where $X$ is reported by strategic agents and $X'$ contains the $\mac$ predictions for $X$}
	\KwOut{The facility location in $\R^d$}
   return $\minbb(X, \textnormal{Best-Choice-Single-Facility-Loc}(X, X', \delta))$
\end{algorithm}

\subsubsection{Supporting small \texorpdfstring{$\eps > 0$}{} values}\label{subsec:the-effect-of-eps}
The above analysis assumed that $\eps = 0$, meaning each prediction is
either wrong, or exactly correct. We can weaken this assumption by
allowing a prediction to be ``correct'' if it is $\eps$-accurate. The
effect of $\eps$ on the approximation ratio is bounded by an additive
term of at most $\eps n$, simply because each of the $n$ points
``contributes'' another additive $\eps$ term.

If $\eps n$ is small in comparison to $OPT$, or equivalently, if
$\eps$ is small compared to $\nf{OPT}{n}$ then the effect on the
approximation ratio will be small. On the other hand, if $\eps$ is big
in comparison to $\nf{OPT}{n}$ then we have no chance to compute a
good solution. Indeed, $\nf{OPT}{n}$ is the \emph{average deviation}
of the points from their true geometric median, which can be viewed as
a notion of the variance of the input; this is a measure of the
``scale'' of the problem, and a noise level above this value swamps
the signal in the data.


\subsection{Balanced \texorpdfstring{$k$}{}-Facility Location in General Metric Spaces}
\label{subsec:balanced-k-fac}


For $\beta \in [0,1]$, the $\beta$-balanced $k$-facility location
problem considers $n$ agents with locations
$X = \{ x_1, \ldots, x_n \} \subseteq V$. Each agent reports a location to the mechanism, and the goal is to return a multi-set
$H = (h_1, \ldots, h_k)$ containing $k$ points from $V$ that minimizes
$\kmed[k](X, H)$, such that the clustering induced by $H$ is
$\beta$-balanced (according to Definition~\ref{def:center-induced}).

For a small enough $\delta$, by choosing $b := \nf{\beta}{\delta}$ the balancedness condition translates into a minimum cluster size of $b \delta$. Hence, we can utilize the $\delta$-robustness results of $\betabalancedkmed[(b-1)\delta]{k}$ (\cref{thm_balanced_k_median_robustness}) to immediately obtain the deterministic mechanism in Algorithm~\ref{alg:balanced_k_facility_loc} with the following guarantee:

\begin{thm} 
\label{thm:blanaced-k-fac}
    \cref{alg:balanced_k_facility_loc} is a deterministic strategyproof mechanism with a constant ($k$-dependent) approximation ratio of at most 
    $$
    1 + \frac{4k}{b-2 - 2k}.
    $$
\end{thm}

\begin{algorithm}[H]
	\SetAlgoNoLine
    \caption{Balanced-k-Facility-Loc($X, X', b, \delta$)}
\label{alg:balanced_k_facility_loc} 
    \KwIn{$X, X' \subseteq V$, and $b > 1$, $\delta \in [0,0.5)$ where $X'$ are the $\mac$ predictions for $X$}
	\KwOut{The $k$ facility locations in $V$}
    return $\betabalancedkmed[(b-1)\delta]{k}(X')$
\end{algorithm}

We do not specify how  $\betabalancedkmed[(b-1)\delta]{k}$ is
implemented. If all points lie on the real line, then 
there exists an $O(n^k)$-time algorithm for it.

In Euclidean spaces of dimensions greater than one,
the $k$-medians problem (and its balanced or lower bounded variant) is classified as NP-hard \cite{megiddo1984complexity,bhattacharya2020hardness}
necessitating the use of approximation algorithms. Any approximation algorithm with approximation ratio $c$ can be applied, incurring an additional multiplicative cost factor of 
$c$. This approach allows for practical solutions within the
constraints of computational complexity, ensuring that we can still
achieve near-optimal placements of facilities even in high-dimensional
contexts. The effect of $\eps$ values and handling the case of high probability $\mac$ predictions is similar to the single facility location results.



\section{Randomized Mechanism Design for \texorpdfstring{$2$}{}-Facility Location on a Line} 
\label{sec:rand-two-fac}

We now consider the $2$-facility-location problem, where the metric is
given by the real line $\RR$. Formally, our metric is over the point
set $V = \RR$, with $d(x,y) = |x-y|$. There are $n$ agents with
locations on the real line $X = \{x_1, \ldots, x_n\}$, and the cost of
a solution $G \in V^2$ is $\kmed(X, G) = \sum_{x_i \in X} d(x_i, G)$.

In the mechanism-design setting, the agent locations $X$ are unknown,
but we are given (a) \mac predictions $X'$ for these locations, as
well as (b) the agents' reports for their locations.
 For simplicity, we assume that $\eps = 0$; the general case can
be handled in a manner similar to that of a single facility.

Unlike in \cref{subsec:balanced-k-fac} we do not assume the balancedness of the optimal solution, but handle the unconstrained case where the optimal solution might induce either a balanced or an unbalanced clustering of $X$. In this section, we present a strategy-proof radnom mechanism with a small (expected) approximation ratio.

Before doing so, we solve two other problems. After we solve those in Sections~\ref{subsec:second-facility-loc-problem} and \ref{sec:big-cluster-center-robustness}, we use the solutions in our mechanism for the $2$-facility-location on a line problem in Section~\ref{subsec:two-fac-alg}.

\subsection{The Second Facility Location Problem}
\label{subsec:second-facility-loc-problem}%

First, we define and solve another independent mechanism design
problem, which we name the \emph{second facility location} problem. We
show a randomized mechanism for this problem that we later use to
solve the $2$-facility-location problem.

\begin{definition}[Second Facility Location]
  Given a metric space $(V,d)$, a dataset $X \in V^n$, and a single
  facility $h_S \in V$, the goal is to find another facility $h_T$ to
  minimize the $2$-medians cost $\kmed[2](X, \{h_S,h_T\})$.
\end{definition}

While we can enumerate over all such second facilities to find the
best solution, the resulting solution is not strategyproof. Indeed,
suppose $X = \crl*{3, 5, 14}$ and $g_S = 0$, then the
agent at $x_1 = 3$ can lower its cost by reporting $x'_1 = 5$ (causing the second facility to be located at $5$ rather than at $14$).
So instead, we consider the
following $\propmech$, which is the second step of the random mechanism for the 2 facility location on the line problem
proposed by \cite{lu2010asymptotically}.

\begin{algorithm}[H]
	\SetAlgoNoLine
    \caption{\propmech}
  \label{alg:second_facility_prop_mech} 
 \KwIn{$X \subseteq \R^d$ a multi-set of $n$ points, and $h_S \in \R^d$ the given first facility location}
	\KwOut{The second facility location in $\R^d$}
    $a_i \gets d(x_i, h_S)$\;
    Pick $h \in X$ such that $Pr[h = x_i] = \nicefrac{a_i}{\sum_{j \in [n]} a_j}$;
\end{algorithm}

\begin{lem}
  \label{lem:prop-is-strategyproof}
  $\propmech$ is strategyproof for any metric space.
\end{lem}

\begin{thm}[Second Facility Results]
  \label{thm:second-facility}
  Consider the case where the metric space is the real line. Fix any
  dataset $X$ and first facility $\{g_S\}$. For any second facility
  $g^*_T \in V$; let $(S,T)$ be a partition of the dataset induced by
  $F^* = \{g_S, g_T^*\}$. The expected cost of $\propmech$ given $X$
  and $g_S$ is:
  \begin{gather}
    \EE[\kmed[2](X,\{g_S, g_T\})] \leq 2 \kmed[1](S,\{g_S\}) + 3
    \kmed[1](T,\{g_T^*\}), 
  \end{gather}
  where $g_T$ is the second facility location chosen by $\propmech$.
\end{thm}

We defer the proofs of \cref{lem:prop-is-strategyproof} and
\cref{thm:second-facility} to \cref{sec:sec-fac-proofs}.
\Cref{thm:second-facility} implies an approximation of a factor of $3$
in general, and this is tight. To show tightness, consider an instance with $n$
points on the real line: $\nf{n}{2}$ points at $x=0$, $\nf{n}{2} - 1$
points at $x=1$, and a single point at $x=2$. Assume the given first
facility is at $0$. The optimal solution puts the second facility at
$1$ and pays the cost of $1$. On the other hand
$\E[ALG] = \frac{1}{\frac{n}{2} + 1} \prn*{\frac{n-2}{2} +
  2\frac{n-2}{2}} = 3 \frac{ \nf{n}{2} -1 }{ \nf{n}{2} + 1} \approx
3$.



\subsection{The \texorpdfstring{$\delta$}{}-Robustness of the Big Cluster Center} 
\label{sec:big-cluster-center-robustness}

We move on to proving the $\delta$ robustness of a location estimator we name $\bcc$. This is the last piece of the puzzle we need for our mechanism for the $2$-facility-location with $\mac$ predictions problem. For this problem, the optimal solution is obtained by $\TwoMed$ if the points are known. Our mechanism uses the predictions to estimate one facility, and randomly chooses the second facility. For this purpose, we come up with an algorithm to ``estimate'' the first facility location.  We formally define what it means to ``estimate'' one of the two facilities, and then we quantify the distance and approximation robustness of the estimator.

Since the breakdown point of $\TwoMed$ is $\nf{1}{n}$, we can not achieve a robust estimation of the optimal solution to the $\TwoMed$ of $n$ points (that is, a solution close to the optimal two centers). However, we can still get a good estimation for one of the two centers of the optimal solution. For any clustering induced by any $2$ points in space, there is always one bigger cluster (that is, the cluster that contains at least half the points). In this section, we show an algorithm, $\bcc$, that ``estimates'' the center of the bigger cluster center out of the two clusters induced by $\TwoMed$.


\begin{definition}[$\bcccost$]
  For any $X = \{ x_1, \ldots, x_n \} \in V^n$, let $F$ be its optimal $2$-median, and let
  $\calC = \calC(X,F)$ be the induced clustering. Let $\bc(X)$ be the
  cluster in clustering $\calC$ with at least $\nf{|X|}2$ points (if both clusters have the same size $\bc(X)$ can be the one containing $x_1$). For
  $t \in V$, define
  \[ \bcccost(X, t) := \kmed[1](\bc(X), t). \]
\end{definition}

\begin{definition}[Big Cluster Location Problem]
  Given $n$ points $X \in V^n$,
  the goal is to find a point $t \in V$ to
  minimize the cost function $\bcccost(X, t)$.
\end{definition}

We specifically focus on the line metric ($V = \R$) and only consider $X \in \R^n$ instances where the clustering induced by $\TwoMed(X)$ are $b \delta$-unbalanced for some $b > 1$, $\delta \in [0,\nf12)$ where $b, \delta$ are some small constants.

\begin{algorithm}[H]
	\SetAlgoNoLine
    \caption{$\bcc$}
    \label{alg:big_cluter_center} 
    \KwIn{$X \subseteq \R$}
	\KwOut{A single facility location in $\R$}
    $(h_L, h_R) = \betabalancedkmed[(b-1)\delta]{2}(X)$ \;
    $L' = \{ x_i \in X \mid d(x_i, h_L) \le d(x_i, h_R) \}$\;
    $ R' = X \setminus L' $\;
    Return $h_L$ if $|L'| \ge |R'|$ and $h_R$ otherwise\;
\end{algorithm}

\begin{thm}
  \label{thm:big_cluster_center_robustness}
  Let $\calY$ be the collection of all datasets for which the optimal
  $2$-medians induces no $b\delta$-balanced clusterings. Then for a small enough $\delta$, Algorithm $\bcc$ is
  $(1.8 + O(\delta), \delta)$-approximation-robust for the cost function $\bcccost$ restricted to instances in $\calY$.
\end{thm}

We given a proof sketch here, and defer the full proof to \cref{sec:robust-bigcluster-proof}.

\begin{proofsketch}
  Let $X \in \calY$, and let $G = \{g_L, g_R\}$ be the optimal $2$-medians for $X$, and let $g_L < g_R$. Let $L,R$ be a clustering of $X$ induced by $G$; by the definition of $\calY$, this clustering is $b\delta$-unbalanced. W.l.o.g., let $|L| > (1-b\delta)n$. let $X'$ be $\mac(0, \delta)$ predictions for $X$. To prove the theorem, we show the following:
  calling $\bcc$ on $X'$
  returns a point $t \in V$ such that
  \begin{gather}
    \med(L,t) \leq (1.8 + O(\delta)) \cdot \med(L,g_L). \label{eq:2}
  \end{gather}
  Let $H = (h_L, h_R) = (b-1)\delta\BalTwoMed(X')$ as in the first
  line of
  \Cref{alg:big_cluter_center}, 
  where $h_L$ lies to the left of $h_R$. Let $L', R'$ be the
  clustering of $X'$ induced by $H$. Let $m = \frac{g_L + g_R}{2}$
  and $m' = \frac{h_L + h_R}{2}$.

In this
  sketch we focus on the case where $h_L \leq g_L$ and $h_R \leq m \leq g_R$. In all other cases, we argue that $L$ and $L'$ share $(1-O(\delta))n$
  points and hence:
  (a) $L'$ is the larger cluster so $\bcc(X')$ returns $h_L$, and (b) since $L$ and $L'$ differ on few points, the approximation-robustness of $\OneMed$ implies
  that $\med(L,h_L) \leq (1+O(\delta))\, \med(L,g_L)$.

 We assume that $m' \ge g_L$. There is a reduction shown from the case of $m' > g_L$ to this case that costs another multiplicative factor of at most $1 + O(\delta)$. This allows us to pin down
  the relative order of the points on the line:
  \[ h_L \leq g_L < m' < h_R \leq m \leq g_R. \] This gives a
  partition of the points in $L$ (which lie to the left of $m$) into
  $S,T,U,V$ as in \cref{figure_big_cluster_approx_hard_case_g_L_le_m_tag_proofsketch} below:
  $S = \{ i \in L \mid x_i \le h_L \}$, $T = \{ i \in L \mid i \notin S, x_i \le m' \}$, $U = \{ i \in L \mid i \notin S \cup T, x_i \le h_R \}$, $V = \{ i \in L \mid i \notin S \cup T \cup U \}$

Let $g_L$ divide $L$ into $L_l$
  and $L_r$, depending on whether they are to the left or right of
  $g_L$.
Similarly $g_R, h_L, h_R$ divides $R, L', R'$ into their left and right parts.

Let $m'' = \frac{g_L + h_R}{2}$ be the middle point between $g_L$ and $h_R$. W.l.o.g we assume that $m'' > m'$ (otherwise the proof is similar). We define $U_l, U_r$ to be the as follows:
$U_l = \{ i \in U \mid x_i \le m'' \}$, $U_r = U \setminus U_l$.

    \begin{figure}[h]
    \centering
    \includegraphics[scale=0.4]{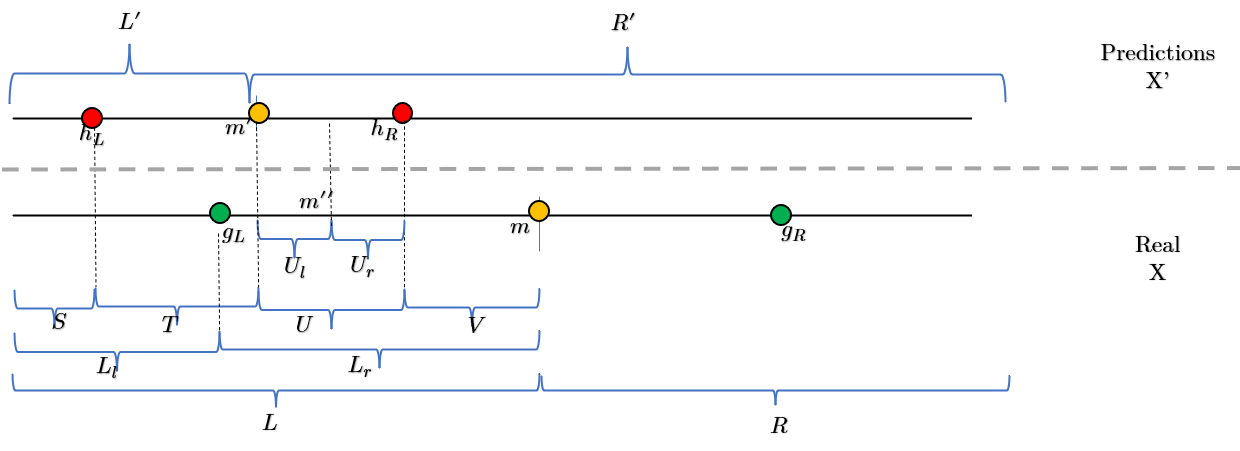}
    \caption{Illustration of case (4) where $m' \ge g_L$. On the top we have $h_L, h_R$, computed on the "predicted" locations $X'$. On the bottom we have the $g_L, g_R$, the $\TwoMed$ of the "real" locations $X$.}
    \label{figure_big_cluster_approx_hard_case_g_L_le_m_tag_proofsketch}
    \end{figure}

  For every $M \subseteq X$ we define $OPT_M$ to be the cost that the optimal solution ($G$) pays for the multi-set of points $M$. That is: $OPT_M := \kmed[2](M, G)$ (we sometimes use this notation where $M \subseteq [n]$ in which case $OPT_M$ is a slight abuse of notation for $OPT_{\{ x_i | i \in M \}}$).

  Let $\beta, \gamma \in [0,1]$ be the ratio between $|S|,|V|$ and $|L|$. That is:
  $\beta := \nf{|S|}{|L|}$ and $\gamma := \nf{|V|}{|L|}$.

  From the definition of $\kmed[1]$, a calculation shows:
  \begin{equation}\label{h_L_cost}
      \kmed[1](L, h_L) = OPT_S - OPT_{T \cap L_l} + OPT_{L_r} + (g_L - h_L)(|L|(1 - 2 \beta)).
  \end{equation}

  If the algorithm returns $h_L$ then we need to show: $\kmed[1](L, h_L) \le (1.8 + O(\delta)) \med(L, g_L)$. The other case where the algorithm returns $h_R$ is omitted from this
  sketch; the proof is similar, and shows a distance bound on
  $h_R-g_L$. From \cref{h_L_cost}, this is equivalent to showing:
  \begin{equation} \label{desired_ineq_dist_robust}
      g_L - h_L \le \frac{0.8 OPT_S + 2.8 OPT_{T \cap L_l} + 0.8 OPT_{L_r} + O(\delta) OPT_L}{|L|(1 - 2 \beta)}.
  \end{equation}
  
  An interpretation of the above is that it is enough to get a distance robustness result for the $\bcc$ estimator which is stated in the above inequality \eqref{desired_ineq_dist_robust} for us to get the desired approximation-robustness result.
  
  To do so, we utilize the following inequalities (We drop their proofs from this sketch):
  \begin{align}
      (g_L - h_L)2 |S| &\le 2 OPT_S, \\
      (g_L - h_L)|U_r| &\le 2 OPT_{U_r}, \\
      (g_L - h_L)2 |V| &\le 2 OPT_{V}, \text{ and }\\
      (g_L - h_L)(|U_l| - O(\delta n)) &\le 2 \prn*{OPT_{T \cap L_l} + OPT_{U_l}}.
  \end{align}
  By summing the above inequalities:
  \begin{align*}
    & (g_L - h_L) \prn*{2 |S| + |U| + 2 |V| - O(\delta n)} \le 2 \prn*{ OPT_S + OPT_{T \cap L_l} + OPT_U + OPT_V}. \implies \\
    & g_L - h_L \le \frac{\prn*{ OPT_S + OPT_{T \cap L_l} + OPT_U +
      OPT_V}}{(1-2\beta)|L|} \cdot \frac{2 (1 - 2\beta)|L|}{2 |S| +
      |U| + 2 |V| - O(\delta
      n)}. \numberthis \label{g_L_minus_h_L_tmp_ineq} 
\end{align*}
Since $|{L'}_l| = |{L'}_r|$, then $|T| = |S| = \beta |L|$ up to an additive factor of at most $\delta n$. Similarly, $|U| = (1 - \gamma)|L| + n$ up to an additive factor of at most $\delta n$. 

By the above approximations of $|T|, |U|$,
the definition of $\beta, \gamma$, and by the fact that $|S| + |T| + |U| + |V| = |L|$ it is possible to show that: 
\[ 2 |S| + |U| + 2|V| - O(\delta n) \ge 2|L| - \beta|L| -\nf{n}{2} -
O(\delta n). \] Plugging this into $\cref{g_L_minus_h_L_tmp_ineq}$ we get:
\begin{align*}
    & g_L - h_L \le \frac{\prn*{ OPT_S + OPT_{T \cap L_l} + OPT_U + OPT_V}}{(1-2\beta)|L|} \cdot \prn*{ \frac{2 (1 - 2\beta)|L|}{2 |L| - \beta |L| - \frac{n}{2} - O(\delta n)}}. \numberthis \label{last_intermediate_result_on_g_L_minus_h_L}
\end{align*}

Analysis of the last expression on RHS shows it is bounded by $\frac{4}{5} + O(\delta)$ and thus we get the desired \cref{desired_ineq_dist_robust}.
\end{proofsketch}

\begin{remark}
    \cref{alg:big_cluter_center} does not have $(c, \delta)$ approximation robustness for any $c < \frac{5}{3} \approx 1.667$. This lower bound is obtained by observing the following instance: Let $\eps > 0$ be a small enough value, $X \in \R^n$ where $\nf{n}{2}$ points are at $x=0$, $\nf{n}{4}$ points are at $x=-0.5 - \eps$, $\nf{n}{4}$ points at $x = -1$, and one point at $M \gg n$. Now let $X'$ be obtained from $X$ by moving the point at $x = M$ to $x = -1$. Calculation shows that the above instance leads to a $\frac{5}{3}$ approximation ratio.
\end{remark}


\subsection{The \texorpdfstring{$2$}{}-Facility Location Algorithm}
\label{subsec:two-fac-alg}

Our randomized algorithm below first uses the \bcc procedure on the
predictions to approximate the ``big'' cluster center, and then uses
the \propmech to approximate the second cluster center using only the
agents' reports. (We set $b \in \mathbb{R}_{> 1}$ such that the approximation robustness in \cref{thm_balanced_k_median_robustness} equals 1.2.)


\begin{algorithm}[H]
	\SetAlgoNoLine
    \caption{Predict-And-Choose-Proportionally-Mechanism}
\label{alg:predict-and-choose-proportionally} 
 \KwIn{$X, X' \subseteq \R^d$ where $X'$ are the $\mac$ predictions for $X$}
 $h_1 = \bcc(X', b, \delta)$\;
 $h_2 = \propmech(X, h_1)$\;
 Return $H = (h_1, h_2)$\;
\end{algorithm}

\begin{thm}
\label{thm:rand-two-fac}
  \cref{alg:predict-and-choose-proportionally} is strategyproof and has an expected approximation ratio of $\ 3.6 + O(\delta)$.
\end{thm}
\begin{proof}
  The strategyproofness is due to the fact that the first facility is
  chosen based only using the predictions, and the second facility
  choice is strategyproof due to \cref{lem:prop-is-strategyproof}.

  Let $G = (g_L, g_R)$
  be the optimal cluster centers for $X$, and let $L, R$ be the
  respective corresponding clusters.  We consider two cases: when this
  optimal clustering is $b\delta$-balanced, and when it is not; in
  both cases we show that the algorithm achieves low expected cost.
    Let $H= (h_1, h_2)$ be the solution returned by
  \cref{alg:predict-and-choose-proportionally}.  
  
    The first case is when $(L, R)$ is a $b \delta$-balanced clustering of $X$.
    In this case we claim that
    \[ \E[\kmed[2](X,H)] \le (3.6 + O(\delta)) \kmed[2](X,G). \]
    Indeed, let $T = (t_L, t_R)$ be the two centers of the
    $(b-1)\delta$-balanced $2$-medians algorithm on $X'$.  Since $G$
    induces a $b \delta$-balanced clustering,
    \cref{thm_balanced_k_median_robustness} and our choice of $b$ ensures that
    $\kmed[2](X,T) \le 1.2 \kmed[2](X,G)$. 
    
    Since \cref{alg:big_cluter_center} returns one of the centers of the $(b-1)\delta$-balanced $2$-medians algorithm on $X'$ as $h_1$, let us assume  w.l.o.g.\ that $\bcc$ returns $h_1 = t_L$.  
    Let $(X_L, X_R)$ and $(X_L', X_R')$ be the clusterings of $X$ and
    $X'$ induced by $T = \{t_L, t_R\}$.  From
    \cref{thm:second-facility}:
    \begin{align*}
      E[\kmed[2](X, H)] &= \E[\kmed[2](X, \{t_L, h_2\})] \le 2 \kmed[1](X_L, t_L) + 3 \kmed[1](X_R, median(X_R)) \\
                        & \smash{\stackrel{(\star)}{\le}} 2 \kmed[1](X_L, t_L) + 3 (1 + O(\delta)) \kmed[1](X_R, t_R) \\
                        & \le (3 + O(\delta)) \kmed[2](X, T) \le (3.6 + O(\delta)) \kmed[2](X,G),
    \end{align*}
    where the inequality $(\star)$ uses the fact that $X_L$ and $X_L'$
    differ on at most $\delta\abs{X}$ points, and hence we can sue
    $(1 + O(\delta) ,\delta)$ approximation-robustness of $\OneMed$
    from \cref{1_median_approx_robustness_results}.

    Now for the other case: suppose $(L, R)$ is a
    $b \delta$-unbalanced clustering.
    \cref{thm:big_cluster_center_robustness} implies that
    either $\kmed[1](L,h_1) \le (1.8 + O(\delta))
      \kmed[1](L, g_L)$ or $\kmed[1](R,h_1) \le (1.8 +
      O(\delta)) \kmed[1](R, g_R)$. W.l.o.g., consider the first
    option. Then from \cref{thm:second-facility}:
    \begin{align*}
      E[\kmed[2](X, H)] 
      & \le 2 \kmed[1](L, h_1) + 3 \kmed[1](R, g_R) \\
      &\le 2 ((1.8 + O(\delta)) \kmed[1](L, g_L)) + 3 \kmed[1](R, g_R) \le \prn*{3.6 + O(\delta)} OPT.
    \end{align*}
  Combining the two cases completes the proof.
\end{proof}

Like for the single facility case, we can modify
\cref{alg:predict-and-choose-proportionally} by combining
$\bcc$ with $\minbb$: this ensures that if there are more than
$\delta n$ prediction errors with probability
$o(\nf{1}{n})$, the approximation ratio would not exceed $O(n)$, and
hence we would still get the same results in expectation.

\section{Conclusion and Future Directions}\label{sec:conclusions}

We explore strategyproof facility location within our introduced $\mac$ predictions error model. 
Our model integrates the trust level into the error model, handles outliers, and leads to new mechanism designs (that do not seem to be
an interpolation between completely trusting the predictions and
resorting to a "no-predictions" method). To analyze our designs, we
utilize distance and approximation robustness notions established in
Section~\ref{sec:robustness}.

An immediate question arising from our results is whether there is a \emph{deterministic} mechanism with $\mac$ predictions for $2$-facility location on a line that yields a constant approximation. 
Another avenue for future research is improving the approximation ratio by randomized mechanisms for this problem. One promising strategy entails refining the selection process for the first facility in \cref{alg:predict-and-choose-proportionally}, which can lead to a superior approximation ratio.
Our current mechanisms treat the multi-sets of predictions and reported values as independent entities; an intriguing direction for further investigation is to devise mechanisms that capitalize on the matching between predictions and agent-reported values, potentially yielding a more accurate approximation ratio.
Lastly, exploring the application of the $\mac$ model to problems beyond mechanism design for 
facility location may yield insights
into the efficacy of such predictions in diverse contexts. 


\section*{Acknowledgements}

We are grateful to Batya Berzack for comments on an earlier draft of the paper.
This work is funded by the European Union (ERC grant No.~101077862, project: ALGOCONTRACT, PI: Inbal Talgam-Cohen), a Google Research Scholar award, the Israel Science Foundation (ISF grant No.~336/18), the Binational Science Foundation (BSF grant No.~2021680), and the National Science Foundation (NSF grant Nos.~CCF-1955785 and CCF-2006953).

\bibliographystyle{ACM-Reference-Format}
\bibliography{references}
\appendix

\newpage
\section{Additional Related Work}
\label{sub:more-related}
\textbf{Mechanism design for facility location problems  without predictions.}

\emph{Mechanism design for the single facility location problem.} In the single facility location in $\R^d$ problem, the task is to return a single facility minimizing the social cost function. For the case of no-predictions, \citet{meir2019strategyproof} shows a deterministic strategyproof mechanism that gets an approximation ratio of $\sqrt{d}$ which is optimal (The optimality for $d = 2$ is shown by \citet{goel2023optimality}). The optimal deterministic strategyproof mechanism is computing the $\cwmed$ of the points. The goal of introducing predictions is to get something better than $\sqrt{d}$ in this context.

\emph{Mechanism design for the two facility location on the line.}
In this variant the goal is to return two  facilities to minimize the social cost, where all points lie on the real line $\R$.
We want to find a strategyproof mechanism that given the agent reported locations returns the two locations for facilities in $\mathbb{R}$, s.t. the social cost is minimal. A deterministic $n-2$ approximation ratio mechanism (called the "Two Extremes" mechanism) was given by \citet{procaccia2013approximate}, and a lower bound of $\nf{n}{2} - 1$ was given by \citet{lu2010asymptotically} which was later improved to (the tight) lower bound of $n-2$  by \citet{fotakis2014power}.
A randomized strategyproof mechanism with an expected approximation ratio of $4$ was given by \citet{lu2010asymptotically} (which also works for any metric space), while the currently best known lower bound for random mechanisms (\cite{lu2009tighter}) is $1.045$.

\emph{$k$-facility location}. For $k > 2$, \citet{fotakis2014power} show that there is no deterministic anonymous strategyproof mechanism with a bounded approximation ratio for $k$-facility location on the line for the general case (not just balanced) for any $k \ge 3$, even for simple instances with $k+1$ agents. Moreover, they show that there do not exist any deterministic strategyproof mechanisms with a bounded approximation ratio for 2-Facility Location on more general metric spaces, which is true even for simple instances with 3 agents located in a star.

\emph{Mechanism design for the capacitated facility location} is a variant of the problem studied by \citet{aziz2020facility} where each facility has a maximum capacity, limiting the number of points that can be assigned to it. For the utilitarian cost function, they have shown a $\nf{n}{2} - 1$ approximation ratio.

\medskip
\textbf{Robust $k$-medians and facility location.}
The robustness of the (offline non-strategic) $k$ medians and $k$ facility location has been studied under different variations. As \cref{example:non-robustness-of-two-median} demonstrates, it is not possible to get any bounded approximation ratio for the optimal solution. 
The approach of \citet{charikar2001algorithms} is to look at different variants of the problems with less restrictive objectives. In one variant they consider, the problem is to place facilities so as to minimize the service cost to any subset of facilities of size at least $p$ for some parameter $p$. Another variant they consider allows denial of service for some of the clients with the additional cost of some penalty for each such denied client. For work in these kind of models see \cite{krishnaswamy2018constant, agrawal2023clustering}.

\medskip
\textbf{Other models of advice/predictions.}
In the recent literature of algorithms with predictions there are existing other models for algorithms with predictions which are different from the models we already discussed. We now mention some of them.

Some other ML advice is to focus on values that are learned via classical PAC learnability and focus on the learnability and analysis of sample complexity bounds. The assumption is that there is some distribution generating the input for the algorithm. Then the approach is to learn the distribution in terms of classic supervised learning. This kind of modeling is detailed in   \cite{anand2020customizing}, \cite{lavastida2020learnable}, \cite{anand2021regression}.

 \cite{diakonikolas2021learning} assumes that the advice is not given deterministically, but that they can access the distributions of the inputs and sample from these: they consider access to the instances of interest, and the goal is to learn a competitive algorithm given access to i.i.d. samples. They provide sample complexity bounds for the underlying learning tasks.

Online variants of the facility location problem with different notions of prediction were studied e.g.~by \cite{almanza2021online, fotakis2021learning, jiang2021online_online_fac_with_pred, azar2022online, gupta2022augmenting, azar2023discretesmoothness, anand2022online}. The online setting of the problem differs significantly from our mechanism design setting: the real points arrive in sequence where each time an irrevocable decision must be made by the online algorithm (unlike the mechanism design setting where all the input is given to the mechanism at once, but is reported by strategic agents).

\medskip
\textbf{Learning the trust parameter.}
In the context of online algorithms with prediction, \citet{khodak2022learning}, show that the confidence parameter can be estimated under certain conditions via online learning. \citet{sun2023online} also have online learning analysis to estimate their different type of confidence parameter.

\medskip
\textbf{Other algorithmic game theory with prediction work}
Research in algorithmic game theory incorporating predictions, beyond the previously discussed work, has been explored by \citet{DBLP:conf/innovations/BalkanskiGT23, DBLP:journals/corr/abs-2310-02879, DBLP:journals/corr/abs-2307-07495, DBLP:journals/corr/abs-2302-14234}.

\medskip
\textbf{Models outside of algorithms with predictions literature.}  One way to view the $\mac$ model is to view the predictions as data with  with corruptions. Designing algorithms to try and handle the corruptions was studied before.  Multi-Arm-Bandit with corruptions is such a setting (\cite{DBLP:journals/corr/abs-1803-09353, DBLP:journals/corr/abs-2002-11650, amir2020prediction}).

\citet{zampetakis2023bayesian} propose the use of robust statistical estimators for strategyproof mechanism design in a \emph{Bayesian} setting. They show how to use a location estimator which is robust to corrupting $\delta$ fraction of the data drawn from some known distribution, to get a strategyproof mechanism for the same location problem. This is conceptually very similar to our result for the single facility case. Since there is no robust estimator for the two-facility problem, their approach cannot apply without making Bayesian assumptions; nonetheless, in our work we show how to use  predictions and agent reports to get improvements on the worst-case approximation guarantees.

\section{Generalized Minimum Bounding-Box Mechanism }\label{sec:min-bb}

In this section we properly define the generalized  $\minbb$ mechanism and prove that for points in $\R^d$ (for any $d \ge 1$) it has a tight approximation ratio of $O(n)$ for the $med_1$ cost function.

For any $y \in \R^d$ let $[y]_j$ denote the $j$'th coordinate of $y$.

\begin{algorithm}[H]
	\SetAlgoNoLine
    \caption{$\minbb$}
\label{alg:min_bounding_box}
    \KwIn{$X \subset \R^d$, $o \in R^d$}
	\KwOut{A facility location in $\R^d$ which is inside the minimum bounding box of $X$}
    \For{$j \in [d]$}{
        $o'_j = \minmaxp\prn*{\prn*{[x_1]_j, \ldots, [x_n]_j}, o_j}$
    }
    return $o'$
\end{algorithm}

\cref{alg:min_bounding_box} simply computes, for each coordinate $j$, the closest point to $o_j$ that is inside the minimum closed interval that contains all of the $j$'th coordinates of the points of $X$ (which is what $\minmaxp$ does, see Mechanism 1 of \cite{agrawal2022learning}). Let us now prove an approximation ratio for the $\kmed[1]$ cost.

\begin{thm} \label{thm_min_bounding_box_linear_approx_ratio}
    \cref{alg:min_bounding_box} is strategyproof and has a tight $O(n)$ approximation ratio for the $\kmed[1]$ cost function.
\end{thm}

\begin{proof}
    Let us assume, w.l.o.g, that the geometric median $g$ is at $g = 0$. Let $B$ be the minimum bounding box of $X$. Formally:
    
    $B = \crl*{ y \in \R^d \mid \forall j \in [d]: y_j \in \brk*{\min_{i \in [n]} {[x_i]_j}, \max_{i \in [n]} {[x_i]_j}} }$.
    Let $h$ be the point returned by \cref{alg:min_bounding_box} for $X \subset \R^d$, $o \in \R^d$ (thus $h$ must be inside of $B$). Let $a_j$ be the side length of the box $B$ in each coordinate.
    As usual, we denote $OPT$ to be the cost of the optimal solution and $ALG$ the cost of the algorithm.
    
    \begin{align*}
        &OPT^2 = \prn*{\sum_{i \in [n]} \norm{x_i}}^2 \ge \sum_{i \in [n]} \sum_{j \in {d}} ([x_i]_j)^2 = \sum_{j \in [d]} \prn*{\sum_{i \in [n]} ([x_i]_j)^2} \stackrel{(\star)}{\ge} \sum_{j \in [d]} \frac{a_j^2}{4} \implies \\
        &OPT \ge \frac{1}{2} \sqrt{\sum_j {a_j^2}} \stackrel{(\star\star)}{\ge} \frac{1}{2} \sqrt{\sum_j {h_j^2}} = \frac{1}{2} \norm{h}. \numberthis \label{upper_bound_bounding_box_eq_h_norm_bound}
    \end{align*}

    $(\star)$ is due to the fact that for any $j \in [d]$ there is some $x_i$ s.t. $\abs{[x_i]_j} \ge \frac{a_j}{2}$ since otherwise we would get that the $j$'th side length is smaller than $a_j$.

    $(\star\star)$ explanation: It is a known property of the geometric median that it lies inside the convex hull of the points.
    Thus, $g$ lies inside $B$ (since the convex hull of the points lies inside the bounding box of the points). So by the fact that $h \in B$ we get that $\abs{h_j} = \abs{h_j - 0} = \abs{h_j - g_j} \le \abs{a_j}$.

    Finally:
    \begin{align*}
        & ALG = \sum_{i \in [n]} \norm{x_i - h} \stackrel{\textit{triangle inequality}}{\le} \sum_{i \in [n]} \norm{x_i} + n \norm{h} \stackrel{\cref{upper_bound_bounding_box_eq_h_norm_bound}}{\le} OPT + 2n \ OPT = O(n) \ OPT.
    \end{align*}

    We deduce that \cref{alg:min_bounding_box} has an approximation ratio of at most $O(n)$.

    To see that the result is tight, consider the instance where $X$ is the multi-set of $n$ points where $n-1$ points are at $a = (-1, \ldots, -1) \in \R^d$ and one point is at $b = (1, \ldots, 1) \in \R^d$. Let $o = b$. The optimal solution puts the facility at $a$ and has a cost of $OPT = d(a,b) = 2 \sqrt{d}$.
    \cref{alg:min_bounding_box} returns $b$ and therefore the cost of the algorithm is $ALG = (n-1) \cdot d(a,b) = (n-1) \cdot 2 \sqrt{d}$.
    
    Overall we get: $\frac{ALG}{OPT} = \frac{(n-1) \cdot 2 \sqrt{d}}{2 \sqrt{d}} = n-1 = \Omega(n)$.

 The strategyproofness of the mechanism is similar to the one proof given by \cite{agrawal2022learning} for the egalitarian cost function.
\end{proof}

One could wonder why use the minimum bounding box rather than the convex hull of the points. After all, the convex hull is always contained in the minimum bounding box, and always contains the geometric median. Unfortunately, such a mechanism is not strategyproof:
\begin{remark}
     The mechanism obtained by replacing the minimum bounding box in \cref{alg:min_bounding_box} with the convex hull is not a strategyproof. The example showing this is for $d=2$: Given $X = \crl*{(-0.5, 0), (0.5, 0), (0, 1)}$, a computation shows that the agent at $(0,1)$ can lower its cost by reporting $(\nf{1}{2}, 1)$.
\end{remark}

We can further generalize the mechanism for $k$ facilities by simply using \cref{alg:min_bounding_box} for each facility independently.
\section{Full Proof of \texorpdfstring{$\bcc$}{}
 \texorpdfstring{$\delta$}{}-Approximation Robustness}\label{sec:robust-bigcluster-proof}

In this section we show that \cref{alg:big_cluter_center} indeed has "good" approximation-robustness for unbalanced clusters by providing the full proof for \cref{thm:big_cluster_center_robustness}.

\begin{proof}
    
Let $G = (g_L, g_R)$ be the $2-median$ solution, where $L = \{i \mid d(x_i, g_L) \le d(x_i, g_R)\}$ and $R = [n] \setminus L$. W.l.o.g $g_L \le g_R$.

We have a slight abuse of notation throughout for the sake of not introducing more variables, referring to $L, R$ as multisets of points and sometimes as the sets of indices of the points in $X$.

We assume that the clusters are $b \delta$-unbalanced. That is, at least one of the clusters $L, R$ are of size less than $b \delta n$. W.l.o.g $|R| < b \delta n$ and so $|L| \ge (1-b\delta) n$.

We also assume that $\delta$ is small enough s.t. $b \delta < \frac{1}{4}$  and so $|L| > \frac{3n}{4}$.


Let $H = (h_L, h_R) = (b-1)\delta\BalTwoMed(X')$ (as in the first line of \Cref{alg:big_cluter_center}).
W.l.o.g $h_L \le h_R$.

Let $L', R'$ be the multi-set of the elements closest to $h_L$ and $h_R$ the remaining. So $L' = \{ i \mid d(x'_i, h_L) \le d(x'_i, h_R) \}$ and $R' = [n] \setminus L'$.

Let $A := \{i \in [n] \mid x_i = x'_i\}$ the shared points between $X, X'$ (the "correct" predictions), and $B := [n] \setminus A$ the remaining.

Let $m = \frac{g_L + g_R}{2}$ be the middle point between $g_L$ and $g_R$. Let $m' = \frac{h_L + h_R}{2}$ be the middle point between $h_L$ and $h_R$. By the definition of $m$: All points of $L$ lie to the left of $m$ and all points of $R$ lie to the right of $m$. Similarly, all points of $L'$ lie to the left of $m'$ and all points of $R'$ lie to the right of $m'$.

We begin by proving the following simple lemma:
\begin{lem}\label{lem_each_center_in_k_median_is_1_median}
    Let $S \sse \mathbb{R}^l$ (for some $l \in \mathbb{N}$) be a multi-set of $n$ points.
    For any $\beta \in [0,1]$, $S \sse V$, let $M = (m_1, \ldots, m_k) := \beta-balanced-k-median(S)$, s.t. $S_1, \ldots, S_k$ are the disjoint clusters induced by $m_1, \ldots, m_k$ of $S$.
    Then for all $i \in [k]$: $m_i$ is the $1-median$ of $S_i$
\end{lem}
\begin{proof}

    Consider the implementation of $\beta-balanced-k-median$ which is obtained by taking the cost-minimizing balanced partition into $k$ $\beta-balanced$ clusters. That is, the implementation considers all of the partitions of $S$ into $k$ $\beta-balanced$ clusters and then computes the location that minimizes the center of each cluster $S_i$. For each such center $S_i$, the minimizing location is (by definition) $1-median(S_i)$. 
    Thus, the implementation always returns points that are the $1-median$ of their clusters.
\end{proof}

We divide the proof into 4 cases as follows (also illustrated in \cref{4_cases_figure}):
\begin{itemize}[leftmargin=+1in]
    \item [Case 1]: $h_L \le g_L$ and $h_R \ge g_R$.
    \item [Case 2]: $h_L \ge g_L$ and $h_R \le g_R$.
    \item [Case 3]: $h_L \ge g_L$ and $h_R \ge g_R$.
    \item [Case 4]: $h_L \le g_L$ and $h_R \le g_R$.
\end{itemize}

Cases 1 to 3 are similar and have short proofs. The hard case is Case (4).

\begin{figure}[h]
    \centering
    \includegraphics[scale=0.4]{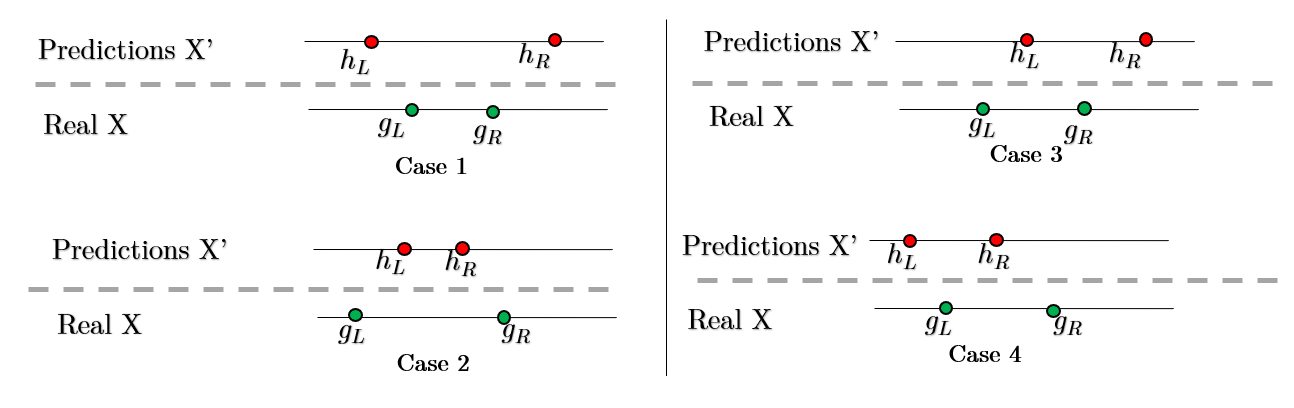}
    \caption{Illustration of the 4 cases. On the top of each case drawing we have the "estimated"/"predicted" locations $H$, computed on $X'$. On the bottom we have the "real" locations, $G$, computed on $X$.}
    \label{4_cases_figure}
\end{figure}

\begin{case}
    $h_L \le g_L$ and $h_R \ge g_R$.
    In this case we can image the two centers moving away from one another. That is, $g_L$ "moves to the left" to $h_L$, and $g_R$ "moves to the right" to $h_R$.

    If $m' \ge m$ then $L'$ contains all of the points in $L \cap A$. and $|L \cap A| \ge |L| - \delta n \ge n - (b+1)\delta n$. Thus, for a small enough $\delta$: $n - (b+1) \delta n \ge \frac{n}{2}$ and thus the algorithm returns $h_L$. From \cref{1_median_approx_robustness_results} we get an approximation-robustness of $1 + O(b\delta) = 1 + O(\delta)$ (since $L'$ is obtained from $L$ by change or add of at most $(b+1)\delta$ elements).

    Otherwise, $m' \le m$. The number of elements to the right of $g_R$ in $R$ is $\frac{|R|}{2}$ since it is the median of $R$ (by \cref{lem_each_center_in_k_median_is_1_median}).
    
    Since $h_R \ge g_R$ then the number of elements to the right of $h_R$ in $R'$ is at most $\frac{|R|}{2} + |B| \le \frac{|R|}{2} + \delta n$.

    $h_R$ is the median of $R'$ (by \cref{lem_each_center_in_k_median_is_1_median}) and thus $|R'| \le 2 (\frac{|R|}{2} + \delta n) = |R| + 2 \delta n \le (b+2)\delta n$. Thus again by \cref{1_median_approx_robustness_results} we get an approximation-robustness of $1 + O(\delta)$.
\end{case}

\begin{case}
     $h_L \ge g_L$ and $h_R \le g_R$. In this case we can image the two centers moving towards one another. That is, $g_L$ "moves to the right" to $h_L$, and $g_R$ "moves to the left" to $h_R$.

     If $m' \ge m$ then $L'$ contains all of $L \cap A$ plus at most $(b+1) \delta n$ points and thus just like the previous case we get (from the 1-median approximation robustness) a approximation-robustness of $1 + O(\delta)$.

     Otherwise $m' < m$. By \cref{lem_each_center_in_k_median_is_1_median} we know that $h_L$ is the median of $L'$ and that $g_L$ is the median of $L$. Since $h_L \ge g_L$, there are at least $\frac{|L|}{2} - \delta n$ points in $L'$ to the left of $h_L$. Thus: $|L'| \ge 2(\frac{|L|}{2} - \delta n) \ge |L| - 2 \delta n \ge (1 - (b+2)\delta) n$, and we can utilize \cref{1_median_approx_robustness_results} again to get the desired.
\end{case}

\begin{case}
    $h_L \ge g_L$ and $h_R \ge g_R$. In this case we can image both centers "moving to the right".

    In this case, $m' = \frac{h_L + h_R}{2} \ge \frac{g_L + g_R}{2} = m$.
    
    Just like in the previous case: $L'$ contains all of $L \cap A$ plus at most $(b+1) \delta n$ points and thus (from 1-median approximation robustness) a approximation-robustness of $1 + O(\delta)$.
    
\end{case}

\begin{case}
    $h_L \le g_L$ and $h_R \le g_R$. In this case we can image both centers "moving to the left".

    By the definition of $m'$: $m' = \frac{h_L + h_R}{2} \le \frac{g_L + g_R}{2} = m$.
    
    Until now we had an approximation-robustness result of $1 + O(\delta)$. This case is the most difficult case, as for this one we will get a $1.8 + O(\delta)$ approximation result. Since we also have a lower bound of $\approx 1.667 + \Omega(\delta)$ we get that this is indeed "strictly" the hard case.

    If $h_R \ge m$ then the number of elements in $X$ to its right is at most $|R|$, and thus the number of elements to its right on $X'$ is at most $|R| + \delta n$. Since $h_R$ median of $R'$ it must be the case that $|R'| \le 2\prn*{|R| + \delta n} \le 2 (b+1)\delta n = O(\delta) \ n$. Thus $|L'| = n - |R'| \ge (1 - O(\delta))n$ and therefore: (a) The algorithm returns $h_L$ (since $|L'| \ge \frac{n}{2}$) and (b) From the $1-median$ robustness (\cref{1_median_approx_robustness_results}) we get the required $1 + O(\delta)$ approximation-robustness. 

    Thus, let us assume that $h_R < m$.

    We first handle the case where $m' \ge g_L$, and then we will show a reduction from the case $m' > g_L$ to this one.

\subsection{Sub-case: \texorpdfstring{$m' \ge g_L$}{}}

    Since $m'$ is the middle point between $h_L$ and $h_R$ we get that: $m' - h_L = h_R - m'$ but $m' - h_L = m' - g_L + g_L - h_L \ge g_L - h_L$ and so: $g_L - h_L \le h_R - m'$. But $h_R - m' \le h_R - g_L$ and therefore:
    \begin{equation}\label{eq_g_L_close_to_h_L}
        g_L - h_L \le h_R - g_L.
    \end{equation}

    We introduce the following notations (also see \cref{figure_big_cluster_approx_hard_case_g_L_le_m_tag}):
    First, we denote the left and right parts of $L', R', L, R$:
    \begin{itemize}
        \item Let ${L'}_l = \{ i \in L' \mid x'_i \le h_L \}, {L'}_r = L' \setminus {L'}_l$.
        \item Let ${R'}_l = \{ i \in R' \mid x'_i \le h_R \}, {R'}_r = R' \setminus {R'}_l$.
        \item Let $L_l = \{ i \in L \mid x_i \le g_L \}, L_r = L \setminus L_l$.
        \item Let $R_l = \{ i \in R \mid x_i \le g_R \}, R_r = R \setminus R_l$.
    \end{itemize}

     Next, we denote the partition of $L$ into 4 disjoint multi-sets $S,T,U,V$ by the 3 points $h_L, m', h_R$.
     So $S = \{ i \in L \mid x_i \le h_L \}$, $T = \{ i \in L \mid i \notin S, x_i \le m' \}$, $U = \{ i \in L \mid i \notin S \cup T, x_i \le h_R \}$, $V = \{ i \in L \mid i \notin S \cup T \cup U \}$.

     Let $m'' = \frac{g_L + h_R}{2}$ be the middle point between $g_L$ and $h_R$. W.l.o.g we assume that $m'' > m'$ (otherwise the proof is similar). We define $U_l, U_r$ to be the as follows:
     $U_l = \{ i \in U \mid x_i \le m'' \}$, $U_r = U \setminus U_l$.

     Finally, we introduce the notion of $\alpha$-approximately-equal: $\approx^\alpha$:
     \begin{definition}
         For any $a, b, \alpha \in \mathbb{R}$: 
         \[
            a \approx^\alpha b \iff b - \alpha \le a \le b + \alpha
        \]
     \end{definition}

    \begin{figure}[h]
    \centering
    \includegraphics[scale=0.42]{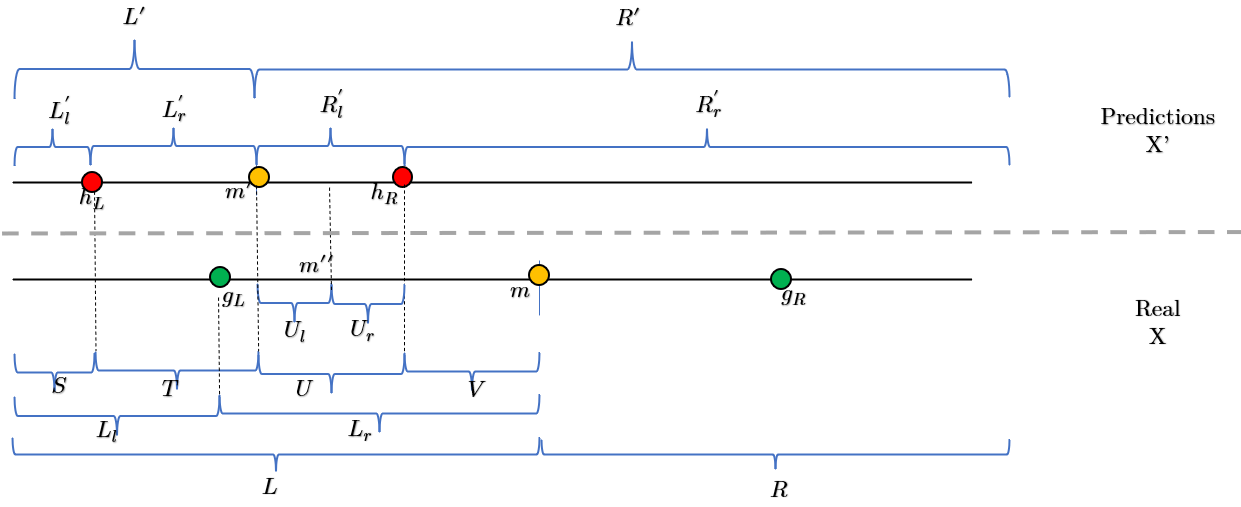}
    \caption{Illustration of case (4) where $m' \ge g_L$. On the top we have $h_L, h_R$, computed on the "predicted" locations $X'$. On the bottom we have the $g_L, g_R$, the $\TwoMed$ of the "real" locations $X$.
    $L'$ and $R'$ are the disjoint partitions of $X'$ into two multi-sets of points: those closer to $h_L$ and those closer to $h_R$ (respectively). Similarly $L$ and $R$ are the disjoint multi-sets of $g_L, g_R$.
    ${L'}_l, {L'}_r$ are the disjoint partitions of $L'$ into two multi-sets: all points to the left of $h_L$ and all points to the right of $h_R$. In a similar manner $R', L, R$ are partitioned into their left and right parts (${R'}_l, {R'}_r, L_l, L_r, R_l, R_r$).
    We can also see $S, T, U ,V$ which is the disjoint partition of $L$ determined by the points $h_L, m', h_R$.
    Finally, we have $U_l$ and $U_r$ which are the left and right parts of $U$.
    }
    \label{figure_big_cluster_approx_hard_case_g_L_le_m_tag}
    \end{figure}
    \newpage

    For every $M \subseteq X$ we define $OPT_M$ to be the cost that the optimal solution ($G$) pays for the multi-set of points $M$. That is: $OPT_M := \kmed[2](M, G)$ (we sometimes use this notation where $M \subseteq [n]$ in which case $OPT_M$ is a slight abuse of notation for $OPT_{\{ x_i | i \in M \}}$). 

    Let $\beta, \gamma \in [0,1]$ be the ratio between $|S|,|V|$ and $|L|$. That is:
    $\beta := \nf{|S|}{|L|}$ and $\gamma := \nf{|V|}{|L|}$.

    \begin{clm}\label{bound_cost_of_h_L_compared_to_cost_of_g_L}
        \[
            \med(L, h_L) = OPT_S - OPT_{T \cap L_l} + OPT_{L_r} + (g_L - h_L)(|L|(1 - 2 \beta)).
        \]
    \end{clm}
    \begin{proof}
        \begin{align*}
            \med(L, h_L) &= \sum_{i \in L} |x_i - h_L| = \sum_{i \in S} h_L - x_i + \sum_{i \in T \cup U \cup V} x_i - h_L \\
                      &\stackrel{(\star)}{=} \sum_{i \in S} g_L - x_i - |S|(g_L - h_L) - \sum_{i \in T \cap L_l} g_L - x_i + |T \cap L_l|(g_L - h_L) \\
                      & + \sum_{i \in (T \cap L_r) \cup U \cup V} x_i - g_L + (|T \cap L_r| + |U| + |V|)(g_L - h_L) \\
                      &\stackrel{(\star \star)}{=} OPT_S - OPT_{T \cap L_l} + OPT_{L_r} + (g_L - h_L) \prn*{|L| - 2 |S|} \\
                      & = OPT_S - OPT_{T \cap L_l} + OPT_{L_r} + (g_L - h_L) \prn*{|L|(1-2\beta)}.
        \end{align*}

    $(\star)$ - triangle inequality

    $(\star \star)$ - due to the fact that $L_r = T \cap L_r \cup U \cup V$ and $|L| = |S| + |U| + |V| + |T|$.
    \end{proof}

    If the algorithm returns $h_L$ then we want to show that $\kmed[1](L, h_L) \le (1.8 + O(\delta)) \kmed[1](L, g_L)$.

    From \cref{bound_cost_of_h_L_compared_to_cost_of_g_L} we get the following equivalent condition:
    \begin{align*}
        & \med(L, h_L) \le (1.8 + O(\delta)) OPT_L \iff \\
        & OPT_S - OPT_{T \cap L_l} + OPT_{L_r} + (g_L - h_L) \prn*{|L|(1-2\beta)} \le (1.8 + O(\delta)) OPT_L \iff \\
        & (g_L - h_L) \prn*{|L|(1-2\beta)} \le 0.8 OPT_S + 2.8 OPT_{T \cap L_l} + 0.8 OPT_{L_r} + O(\delta) OPT_L \iff \\
        & (g_L - h_L) \le \frac{0.8 OPT_S + 2.8 OPT_{T \cap L_l} + 0.8 OPT_{L_r} + O(\delta) OPT_L}{|L|(1 - 2 \beta)}. \numberthis \label{ineq_we_need_to_show_to_get_delta_approx_results_foR'_L}
    \end{align*}
    
     Thus, all we have to do is find the above bound for $g_L - h_L$. Indeed we will show that this is the case. We start by proving a few claims that will help us bound $g_L - h_L$.

    \begin{clm} \label{bound_on_beta_plus_gamma_L}
        $(\beta+\gamma)|L| \approx^{O(\delta n)} |L| - \frac{n}{2}$
    \end{clm}
    \begin{proof}
        Since $h_L$ is the median of $L'$, and the points in $S\cup T$ differ from those in $L'$ by at most $O(\delta n)$, we get that $|T|
          \approx^{O(\delta n)} \beta |L|$.
        Similarly, since $h_R$ is the median of $R'$ and since $|R| = n - |L|$: $|U| \approx^{O(\delta n)} (\gamma - 1)|L| + n$.
        By using the fact that $|L| = |S| + |T| + |U| + |V|$ we get: 
        \begin{align*}
            & |L| \approx^{O(\delta n)} \beta |L| + \beta |L| + (\gamma - 1)|L| + n + \gamma |L| \implies \\
            & 2 (\beta + \gamma) |L| \approx^{O(\delta n)} 2|L| - n \implies \\
            & (\beta+\gamma)|L| \approx^{O(\delta n)} |L| - \frac{n}{2}.
        \end{align*}
    \end{proof}

    \begin{clm}\label{bound_opt_s}
        \begin{equation}
            (g_L - h_L)|S| \le OPT_S.
        \end{equation}
    \end{clm}
    \begin{proof}
        \[ OPT_S = \sum_{i \in S} g_L - x_i \stackrel{\forall x_i \in S\textit{: } x_i \le h_L}{\ge} \sum_{i \in S} g_L - h_L = |S|(g_L - h_L).
        \]
    \end{proof}

    \begin{clm}\label{bound_opt_U_r}
        \begin{equation} 
            (h_R - g_L)|U_r| \le 2 OPT_{U_r}.
        \end{equation}
    \end{clm}
    \begin{proof}
        \[ OPT_{U_r}= \sum_{i \in U_r} x_i - g_L \stackrel{\forall x_i \in U_r\textit{: } x_i \ge m''}{\ge} \sum_{i \in U_r} m'' - g_L = |U_r|(m'' - g_L) = \frac{|U_r|}{2} (h_R - g_L).
        \]

        where the last equality is due to the fact that $m''$ is the middle point between $g_L$ and $h_R$.

        Thus we get:

        $(h_R - g_L)|U_r| \le 2 OPT_{U_r}$
    \end{proof}

    \begin{clm}\label{bound_opt_v}
        \begin{equation}
            (h_R - g_L)|V| \le OPT_V.
        \end{equation}
    \end{clm}
    \begin{proof}
        \[ OPT_V = \sum_{i \in V} x_i - g_L \stackrel{\forall x_i \in V\textit{: } x_i \ge h_R}{\ge} \sum_{i \in S} h_R - g_L = |V|(h_R - g_L).
        \]
    \end{proof}

Consider the clustering induced by $H' = (g_L, h_R)$.
If this clustering is unbalanced then \cref{alg:big_cluter_center} returns $h_L$ since $L_l \cap A$ is contained in the left cluster and since $|L_l \cap A| \ge \frac{L}{2} - \delta n \ge \frac{n}{2} - O (\delta n)$ it must be that the right cluster is the smaller one. And so the left cluster has at least $n - (b-1)\delta n = n - O(\delta n)$ elements and thus it is obtained from $L$ by modifying or dropping at mots $O(\delta n)$ elements and thus the algorithm returns $h_L$ and like before we get a $(1 + O(\delta), \delta)$ approximation-robustness.

Otherwise, the clustering induced by $H'$ is balanced, and we get the following claim:
 
\begin{clm} \label{bound_opt_u_l}
    $(h_R - g_L)(|U_l| - O(\delta n)) \le 2 (OPT_{T \cap L_l} + OPT_{U_l})$.
\end{clm}
\begin{proof}
    
    For convenience, for any multi-set of the real points $M \subseteq X$ we denote $\hat{M} \subseteq X'$ to be the multi-set of the estimated points $X'$ that correspond to the same partition of the real line as $M$.
    So we define the following multi-sets: $\hat{L_l} = \{ x'_i \in X' \mid x'_i \le g_L \}$, $\hat{L_r} = \{  x'_i \in X' \setminus \hat{L_l} \mid x'_i \le m \}$, $\hat{S} = {L'}_l$, $\hat{T} = {L'}_r$, $\hat{U} = {R'}_l$, $\hat{V} = \{ i \in {R'}_r \mid x'_i \in [h_R, m] \}$,
    $\hat{U_l} = \{ x'_i \in {R'}_l \mid x'_i \le m'' \}$, $\hat{U_r} = \{ x'_i \in {R'}_l \mid x'_i > m'' \}$.
    The reason we use this notation is that we have an inequality in terms of the points in $X'$ and we want to later move on to the inequality in terms of the points in $X$. This notation will help us see the connection between the points in $X'$ and the points in $X$.
    Each such $M, \hat{M}$ contain the same points up to at most $\delta n$ points.

    By definition of $H$ we know that $\kmed[2](X', H) \le \kmed[2](X', T)$ for any $T = (t_1, t_2) \in \mathbb{R}^2$ that induces a  $(b-1)\delta$-balanced clustering of $X'$.
    Therefore, $\kmed[2](X', H) \le \kmed[2](X', (g_L, h_R))$.

    By the definition of $\kmed[2]$:
    \begin{align*}
        & \sum_{x'_i \in {L'}_l} {h_L - x'_i} + \sum_{x'_i \in {L'}_r} {x'_i - h_L} + \sum_{x'_i \in {R'}_l} {h_R - x'_i} + \sum_{x'_i \in R'_r} {x'_i - h_R} \le \\
        & \sum_{x'_i \in \hat{S} \cup (\hat{T} \cap \hat{L_l})} {g_L - x'_i} + \sum_{x'_i \in (\hat{T} \cap \hat{L_r}) \cup \hat{U_l} } {x'_i - g_L} + \sum_{x'_i \in \hat{U_r}} {h_R - x'_i} + \sum_{x'_i \in R'_r} {x'_i - h_R}.
    \end{align*}
    Which implies:
    \begin{align*}
        & \sum_{x'_i \in L_{h_l}} {h_L - x'_i} + \sum_{x'_i \in {L'}_r} {x'_i - h_L} + \sum_{x'_i \in {R'}_l} {h_R - x'_i}
            \\
            & \le
                \sum_{x'_i \in \hat{S} \cup (\hat{T} \cap \hat{L_l})} {g_L - x'_i} + \sum_{x'_i \in (\hat{T} \cap \hat{L_r}) \cup \hat{U_l} } {x'_i - g_L} + \sum_{x'_i \in \hat{U_r}} {h_R - x'_i} \numberthis \label{intermed_res_1}.
    \end{align*}                

    Let us observe LHS:
    \begin{align*}
        &
            \sum_{x'_i \in L_{h_l}} {h_L - x'_i} + \sum_{x'_i \in {L'}_r} {x'_i - h_L} + \sum_{x'_i \in {R'}_l} {h_R - x'_i} \\
        & = |{L'}_l|(h_L - g_L) + \sum_{x'_i \in L_{h_l}} {g_L - x'_i} + \sum_{x'_i \in {L'}_r} {x'_i - g_L} + |{L'}_r| (g_L - h_L) + |{R'}_l|(h_R - g_L) + \sum_{x'_i \in {R'}_l} {g_L - x'_i} \\
        & \stackrel{|{L'}_l| = |{L'}_r| }{=} \sum_{x'_i \in L_{h_l}} {g_L - x'_i} + \sum_{x'_i \in {L'}_r} {x'_i - g_L} + |{R'}_l|(h_R - g_L) + \sum_{x'_i \in {R'}_l} {g_L - x'_i} \\
        & = \sum_{x'_i \in \hat{S}} {g_L - x'_i} + \sum_{x'_i \in \hat{T}} {x'_i - g_L} + |\hat{U}|(h_R - g_L) + \sum_{x'_i \in \hat{U}} {g_L - x'_i} \\
        & = \sum_{x'_i \in \hat{S}} {g_L - x'_i} + \sum_{x'_i \in \hat{T} \cap \hat{L_l}} {x'_i - g_L} + \sum_{x'_i \in \hat{T} \cap \hat{L_r}} {x'_i - g_L} + |\hat{U}|(h_R - g_L) + \sum_{x'_i \in \hat{U}} {g_L - x'_i}.
    \end{align*}

So by plugging this in (\ref{intermed_res_1}) we get:
\begin{align*}
        &\sum_{x'_i \in \hat{S}} {g_L - x'_i} + \sum_{x'_i \in \hat{T} \cap \hat{L_l}} {x'_i - g_L} + \sum_{x'_i \in \hat{T} \cap \hat{L_r}} {x'_i - g_L} + |\hat{U}|(h_R - g_L) + \sum_{x'_i \in \hat{U}} {g_L - x'_i} \\
    & \le 
        \sum_{x'_i \in \hat{S} \cup (\hat{T} \cap \hat{L_l})} {g_L - x'_i} + \sum_{x'_i \in (\hat{T} \cap \hat{L_r}) \cup \hat{U_l} } {x'_i - g_L} + \sum_{x'_i \in \hat{U_r}} {h_R - x'_i}. \\
    & \implies \\
        & |\hat{U}|(h_R - g_L) + \sum_{x'_i \in \hat{U}} {g_L - x'_i}
    \le
        2 \sum_{x'_i \in \hat{T} \cap \hat{L_l}} {g_L - x'_i} + \sum_{x'_i \in \hat{U_l} } {x'_i - g_L} + \sum_{x'_i \in \hat{U_r}} {h_R - x'_i} \\
    & = 2 \sum_{x'_i \in \hat{T} \cap \hat{L_l}} {g_L - x'_i} + \sum_{x'_i \in \hat{U_l} } {x'_i - g_L} + |\hat{U_r}|(h_R - g_L) + \sum_{x'_i \in \hat{U_r}} {g_L - x'_i} \implies \\
    & |\hat{U_l}|(h_R - g_L) \le 2 \sum_{x'_i \in \hat{T} \cap \hat{L_l}} {g_L - x'_i} + 2 \sum_{x'_i \in \hat{U_l} } {x'_i - g_L} \numberthis \label{res_hat}.
\end{align*}

By the definition of $\hat{T}, \hat{L_l}$:  $\hat{T} \cap \hat{L_l}$  and $T \cap L_l$ differ by at most $\delta n$ elements, and for any $x'_i \in \hat{T} \cap \hat{L_l}$: $x'_i \ge h_L$. Similarly $\hat{U_l}$ and $U_l$ differ by at most $\delta n$ elements and for any $x'_i \in \hat{U_l}$: $x'_i \le m''$.

By plugging this in the above (\ref{res_hat}) we get:

\begin{align*}
    (|U_l| - \delta n)(h_R - g_L)
    & \le |\hat{U_l}|(h_R - g_L) \\
    & \le 2 \sum_{x_i \in T \cap L_l} {g_L - x_i} + 2 \delta n (g_L - h_L) + 2 \sum_{x_i \in U_l } {x_i - g_L} + 2 \delta n (m'' - g_L).
\end{align*}
Which implies:
\[
      (h_R - g_L)|U_l| \le 2 (OPT_{T \cap L_l} + OPT_{U_l}) + \delta n \prn*{ 2(g_L - h_L) + 2 (m'' - g_L) + (h_R - g_L)}.
\]

Since $m'' - g_L = \nf{(h_R - g_L)}{2}$ (by the definition of $m''$) and from \cref{eq_g_L_close_to_h_L} we get:

\[
      (h_R - g_L)|U_l| \le 2 (OPT_{T \cap L_l} + OPT_{U_l}) + O(\delta n) (h_R - g_L).
\]

\begin{equation} \label{bound_opt_U_l_and_t_cap_L_l}
    (h_R - g_L)(|U_l| - O(\delta n)) \le 2 (OPT_{T \cap L_l} + OPT_{U_l}).
\end{equation}
\end{proof}

Now we are ready to bound $g_L - h_L$.

\begin{lem}\label{lem_bound_dist_between_g_L_and_h_L}
    If the algorithm returns $h_L$ then:
    \[
        (g_L - h_L) \le \frac{0.8 OPT_S + 2.8 OPT_{T \cap L_l} + 0.8 OPT_{L_r} + O(\delta) OPT_L}{|L|(1 - 2 \beta)}.
    \]
\end{lem}

\begin{proof}
    By \cref{bound_opt_s}, \cref{bound_opt_U_r}, \cref{bound_opt_u_l}, \cref{bound_opt_v} and the fact that $g_L - h_L \le h_R - g_L$ we get the following four inequalities:
    \[
        (g_L - h_L)|U_r| \le 2 OPT_{U_r} ,
    \]
    \[
        (g_L - h_L)(|U_l| - O(\delta n)) \le 2 (OPT_{T \cap L_l} + OPT_{U_l}) ,
    \]
    \[
        (g_L - h_L)2|V| \le 2 OPT_V, 
    \]
    \[
        (g_L - h_L)2|S| \le 2 OPT_S, 
    \]

    By summing the above inequalities:
    \begin{align*}
        & (g_L - h_L) \prn*{2 |S| + |U| + 2 |V| - O(\delta n)} \le 2 \prn*{ OPT_S + OPT_{T \cap L_l} + OPT_U + OPT_V} \implies \\
        & g_L - h_L \le \frac{\prn*{ OPT_S + OPT_{T \cap L_l} + OPT_U + OPT_V}}{(1-2\beta)|L|} \cdot \frac{2 (1 - 2\beta)|L|}{2 |S| + |U| + 2 |V| - O(\delta n)}. \numberthis \label{intermediate_result_on_g_L_minus_h_L}
    \end{align*}

    Since $|{L'}_l| = |{L'}_r|$ we get that $|T| \approx^{\delta n} |S| = \beta |L|$ and so
    \begin{align*}
        2 |S| + |U| + 2 |V| - O(\delta n) \\
        & \ge  |S| + |T| + |U| + 2 |V| - O(\delta n) \\
        & = (|S| + |T| + |U| + |V|) + (|S| + |V|) - |S| - O(\delta n) \\
        & \stackrel{(\star)}{=} |L| + (\beta + \gamma)|L| - \beta |L| - O(\delta n) \\
        & \stackrel{\cref{bound_on_beta_plus_gamma_L}}{\ge} |L| + |L| - \frac{n}{2} - \beta |L| - O(\delta n) \\
        & = 2 |L| - \beta |L| - \frac{n}{2} - O(\delta n).
    \end{align*}

    Where $(\star)$ is due to: $|L| = |S| + |T| + |U| + |V|$, $|V| = \gamma |L|$, $|S| = \beta |L|$.

    By plugging this inequality into \cref{intermediate_result_on_g_L_minus_h_L}:
    \begin{align*}
        & g_L - h_L \le \frac{\prn*{ OPT_S + OPT_{T \cap L_l} + OPT_U + OPT_V}}{(1-2\beta)|L|} \cdot \prn*{ \frac{2 (1 - 2\beta)|L|}{2 |L| - \beta |L| - \frac{n}{2} - O(\delta n)}}. \numberthis \label{second_intermediate_result_on_g_L_minus_h_L}
    \end{align*}

    We will bound the second term in RHS.
    First, we know that $|L| \approx^{O(\delta n)} n$ and thus:
    \[
        \frac{2 (1 - 2\beta)|L|}{2 |L| - \beta |L| - \frac{n}{2} - O(\delta n)} \le \frac{2(1 - 2\beta)(n - O(\delta n))}{\frac{3n}{2} - \beta n - O(\delta n)} \le \frac{2(1-2\beta)}{\frac{3}{2} - \beta} \prn*{ 1 + O(\delta)}
    \]

    We denote: $f(\beta) := \frac{2(1-2\beta)}{\frac{3}{2} - \beta}$.

    Since the algorithm returns $h_L$ we deduce that $|L'| \ge \frac{n}{2}$. Thus, $|{L'}_l|, |{L'}_r| \ge \frac{n}{4}$. But $S$ and ${L'}_l$ differ by at most $\delta n$ elements, and so $|S| \ge \frac{n}{4} - \delta n$. Since $|S| = \beta |L|$ and $|L| \le n$ we get: $\beta n \ge \frac{n}{4} - \delta n$ which implies $\beta \ge \frac{1}{4} - \delta$.

    Also obviously $\beta \le 0.5$ since $h_L \le g_L$ by the definition of $S$.

    So $f: [\frac{1}{4} - \delta, \frac{1}{2}] \rightarrow \mathbb{R}$ is a well defined function. We find its maximum:
    
    \[
        f'(\beta) = 2 \frac{-2(\frac{3}{2} - \beta) -(1 - 2\beta)(-1)}{(\frac{3}{2}-\beta)^2} = 2 \frac{-3 + 2\beta + 1 - 2\beta}{(\frac{3}{2}-\beta)^2} = \frac{-4}{(\frac{3}{2}-\beta)^2} < 0.
    \]

    So $f$ gets its maximum at the left boundary where $\beta = \frac{1}{4} - \delta$:

    \[
        f(\frac{1}{4} - \delta) = 2 \frac{\frac{1}{2} + 2\delta}{\frac{5}{4} + \delta} = \frac{4}{5} (1 + O(\delta)).
    \]

    By plugging the above into \cref{second_intermediate_result_on_g_L_minus_h_L} we get:

    \[
        g_L - h_L \le \frac{4}{5} \frac{\prn*{ OPT_S + OPT_{T \cap L_l} + OPT_U + OPT_V + O(\delta) OPT_L}}{(1-2\beta)|L|}.
    \]
    \end{proof}

    What we have shown is that if the algorithm returns $h_L$ we are done. If the algorithm returns $h_R$ we can similarly show that in this case $\kmed[1](L, h_R) \le (1.8 + \Theta(\delta)) \kmed[1](L, g_L)$.
    
    \subsubsection{Handling the case where the algorithm returns $h_R$}
    In this case the cost $\med(L, h_R)$ will be:

    \begin{clm} \label{claim_showing_relation_between_cost_of_h_R_and_dist_between_h_R_and_g_L}
        \[
            \med(L, h_R) = OPT_S + OPT_{T \cap L_l} - OPT_{T \cap L_r} - OPT_U + OPT_V + (|L|(1 - 2\gamma))(h_R - g_L).
        \]
    \end{clm}

    \begin{proof}
    \begin{align*}
        \med(L, h_R)
        & = \sum_{i \in S \cup T \cup U} h_R - x_i + \sum_{i \in V} x_i - h_R\\
        & = \sum_{i \in S \cup T \cup U} (h_R - g_L) + (g_L - x_i) + \sum_{i \in V} (x_i - g_L) - (h_R - g_L) \\
        & = OPT_S + OPT_{T \cap L_l} - OPT_{T \cap L_r} - OPT_U + OPT_V + (|S| + |T| + |U| - |V|)(h_R - g_L) \\
        & \stackrel{|S|+|T|+|U|+|V|=|L|}{=} OPT_S + OPT_{T \cap L_l} - OPT_{T \cap L_r} - OPT_U + OPT_V + (|L| - 2|V|)(h_R - g_L) \\
        & = OPT_S + OPT_{T \cap L_l} - OPT_{T \cap L_r} - OPT_U + OPT_V + (|L|(1 - 2 \gamma))(h_R - g_L).
    \end{align*}
    \end{proof}

    To get $\med(L, h_R) \le (1.8 + O(\delta)) \med(L, g_L)$ we need to show (by the above claim) the following lemma:

\begin{lem}\label{goal_d_h_R_and_g_L}
    If the algorithm returns $h_R$ then:
    \[
         (h_R - g_L) \le \frac{0.8\prn*{OPT_S + OPT_{T \cap L_l} + OPT_V} + 2.8 \prn*{OPT_{T \cap L_r} + OPT_U} + O(\delta)OPT_L}{(1 - 2 \gamma) |L|}.
    \]
\end{lem}
\begin{proof}

    We begin by showing the following claim:
    \begin{clm} \label{bound_h_R_minus_g_L_times_S}
        \[
            (h_R - g_L)|S| \le (h_R - g_L)|S| \le OPT_S + \frac{|S|}{|U|} OPT_U.
        \]
    \end{clm}
    \begin{proof}
        $h_R - g_L = m' - h_L = (m' - g_L) + (g_L - h_L)$ and so $g_L - h_L = (h_R - g_L) - (m' - g_L)$. Together with \cref{bound_opt_s} we get:

    \begin{equation} \label{tmp_bound_1}
        (h_R - g_L)|S| \le OPT_S + |S|(m' - g_L).
    \end{equation}

    Also, $OPT_U = \sum_{i \in U} x_i - g_L \ge \sum_{i \in U} m' - g_L = |U|(m' - g_L)$
    which implies 
    \begin{equation} \label{tmp_bound_2}
        m' - g_L \le \frac{OPT_U}{|U|}.
    \end{equation}
    By \cref{tmp_bound_1} and \cref{tmp_bound_2} we get the desired inequality.
    \end{proof}

    From \cref{bound_h_R_minus_g_L_times_S}, \cref{bound_opt_U_r}, \cref{bound_opt_u_l}, \cref{bound_opt_v}:

    we get the following four inequalities:
    \[
        (h_R - g_L)2|S| \le 2 OPT_S + 2 \frac{|S|}{|U|} OPT_U, 
    \]
    \[
        (h_R - g_L)|U_r| \le 2 OPT_{U_r} ,
    \]
    \[
        (h_R - g_L)(|U_l| - O(\delta n)) \le 2 (OPT_{T \cap L_l} + OPT_{U_l}) ,
    \]
    \[
        (h_R - g_L)2|V| \le 2 OPT_V, 
    \]

    By summing these inequalities:
    \begin{align} \label{intermediate_result_h_R_minus_g_L}
        (h_R - g_L) \prn*{2 |S| + |U| + 2 |V| - O(\delta n)} \le \frac{OPT_S + OPT_U + OPT_{T \cap L_l} + OPT_V + \frac{2|S|}{|U|} OPT_U}{(1 - 2 \gamma)|L|} \cdot 2 (1 - 2 \gamma)|L|.
    \end{align}

    Since the algorithm returns $h_R$ it must be the case that $|R'| \ge \frac{n}{2}$ and so $|{R'}_l| \ge \frac{n}{4}$ which means that $|U| \ge \frac{n}{4} - \delta n$. This also means that $|L'| \le \frac{n}{2}$ and thus similarly $|S| \le \frac{n}{4} + \delta n$.
    So together:
    
    \[
        2 \frac{|S|}{|U|} \le \frac{\frac{n}{2} + 2 \delta n}{\frac{n}{4} - \delta n} = 2  + O(\delta).
    \]

    By plugging this back into \cref{intermediate_result_h_R_minus_g_L} we get:
    \begin{align} \label{second_intermediate_result_h_R_minus_g_L}
        (h_R - g_L) \prn*{2 |S| + |U| + 2 |V| - O(\delta n)} \le \frac{OPT_S + (3 + O(\delta)) OPT_U + OPT_{T \cap L_l} + OPT_V}{(1 - 2 \gamma)|L|} \cdot 2 (1 - 2 \gamma)|L|.
    \end{align}
    
    Since $|{L'}_l| = |{L'}_r|$ we get that $|T| \approx^{\delta n} |S|$ and so
    \begin{align*}
        2 |S| + |U| + 2 |V| - O(\delta n) \\
        & \ge  |S| + |T| + |U| + 2 |V| - O(\delta n) \\
        & = (|S| + |T| + |U| + |V|) + |V| - O(\delta n) \\
        & \stackrel{(\star)}{=} (1+\gamma)|L| - O(\delta n)
    \end{align*}

    Where $(\star)$ is due to: $|L| = |S| + |T| + |U| + |V|$, $|V| = \gamma |L|$.

    Let us use this fact in \cref{second_intermediate_result_h_R_minus_g_L} and get:

    \begin{align} \label{third_intermediate_result_h_R_minus_g_L}
        h_R - g_L \le \frac{OPT_S + (3 + O(\delta)) OPT_U + OPT_{T \cap L_l} + OPT_V}{(1 - 2 \gamma)|L|} \cdot \prn*{2 \frac{(1 - 2 \gamma)|L|}{(1+\gamma)|L| - O(\delta n)}}.
    \end{align}

    The second term on RHS is (up to $1 + O(\delta)$ factor):
    $h(\gamma) := 2 \frac{1-2\gamma}{1+\gamma}$.
    Let us find a bound f or $h$.
    
    Since $|R'| \ge \frac{n}{2}$:
    \[
        |V| + |R| \approx{=}^{O(\delta n)} |{R'}_r| \ge \frac{n}{4}.
    \]
    and since $|R| \le b \delta n$:
    $\gamma |L| = |V| \ge \frac{n}{4} - b \delta n$. We deduce:
    
    $\gamma \ge \frac{1}{4} - O(\delta)$.
    The nominator is decreasing in $\gamma$ and the denominator is increasing in $\gamma$ and thus the maximum is received on the left end point: $\gamma = \frac{1}{4} - O(\delta)$.
    The maximum value of $h$ is thus: $2 \frac{\frac{1}{2} - O(\delta)}{\frac{5}{4} - O(\delta)} = \frac{4}{5} + O(\delta)$.

    So we found a bound for the RHS of \cref{third_intermediate_result_h_R_minus_g_L} and therefore:

    \[
        h_R - g_L \le \frac{\frac{4}{5} OPT_S + \frac{12}{5} OPT_U + \frac{4}{5} OPT_{T \cap L_l} + \frac{4}{5} OPT_V + O(\delta) OPT_L}{(1 - 2 \gamma)|L|}.
    \]
\end{proof}

\subsection{Sub-case: \texorpdfstring{$m' < g_L$}{}.}
To handle this sub-case we use a reduction to the previous sub-case and show that we pay another multiplicative factor of at most $1 + O(\delta)$. We show a bound of $\kmed[1](h_L, L)$ in case the algorithm returns $h_L$. In the case the algorithm returns $h_R$ a bound for $\kmed[1](h_R, L)$ is achieved similarly.

Let $\tilde{g_L} = m'$.
For any multi-set of points $M \subseteq L$ let $\widetilde{OPT_M} = \sum_{i \in M} |x_i - \tilde{g_L}|$.

\begin{figure}[h]
    \centering
    \includegraphics[scale=0.45]{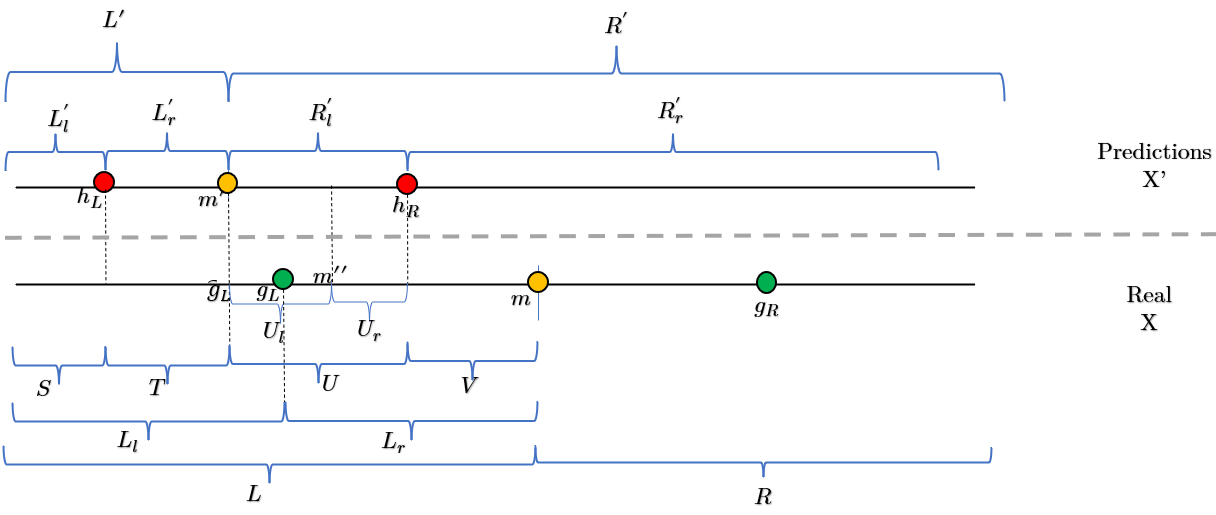}
    \caption{Illustration of case (4) where $m' < g_L$. On the top we have the "estimated" locations $h_L, h_R$, computed on $X'$. On the bottom we have the "real" locations, $g_L, g_R$, computed on $X$.
    $L', R', {L'}_l, {L'}_r, {R'}_l, {R'}_r, S, T, U ,V, U_l, U_r$ are all the same as before.
    The difference is that $m'$ is now to the left of $g_L$.
    }
    \label{figure_big_cluster_approx_hard_case_g_L_bigger_than_m'}
    \end{figure}
    \newpage

By the exact same analysis of the case where $m' \ge g_L$, now we have that $m' \ge \tilde{g_L}$ and thus, just like before, 
$\tilde{g_L} - h_L \le \frac{4}{5} \frac{\widetilde{OPT_S} + \widetilde{OPT_{T}} + \widetilde{OPT_{L_r}}}{(1 - 2\beta) |L|}$.

But $\widetilde{OPT_S} + \widetilde{OPT_T} = OPT_S + OPT_T - (g_L - m')(|S| + |T|)$ and 
$\widetilde{OPT_{L_r}} = OPT_{L_r} + (g_L - m')|L_r|$.

So together we have:

\begin{align} \label{mh_smaller_than_g_L_bound_on_tilde_diff}
    \tilde{g_L} - h_L &\le \frac{4}{5} \frac{OPT_S + OPT_T - (g_L - m')(|S| + |T|) + OPT_{L_r} + (g_L - m')(|U| + |V|)}{(1 - 2\beta) |L|} \\
    & = \frac{4}{5} \frac{\prn*{OPT_S + OPT_T + OPT_{L_r} + (g_L - m')(-|S|-|T| + |U|+|V|)}}{(1 - 2\beta) |L|}.
\end{align}

Since $|L'| \ge \frac{n}{2}$ it must be the case that $|S|+|T| \ge \frac{n}{2} - \delta n$ since otherwise there would be strictly less than $\frac{n}{2}$ elements in $|L'|$. For this reason also $|U| + |V| \le \frac{n}{2}$

So $-|S|-|T| + |U| + |V| \le \delta n $. By plugging this in \cref{mh_smaller_than_g_L_bound_on_tilde_diff} we get:

\begin{align*}
    \tilde{g_L} - h_L &\le \frac{4}{5} \frac{\prn*{OPT_S + OPT_T + OPT_{L_r} + (g_L - m')(\delta n}}{(1 - 2\beta) |L|}.
\end{align*}

Note that $\tilde{g_L} - h_L = \tilde{g_L} - g_L + g_L - h_L = g_L - h_L - (g_L - \tilde{g_L}) = (g_L - h_L) - (g_L - m')$.

By plugging this in the above we get:

\begin{align*}
    g_L - h_L &\le \frac{4}{5} \frac{\prn*{OPT_S + OPT_T + OPT_{L_r} + \frac{10}{4}(g_L - m') \delta n}}{(1 - 2\beta) |L|}.
\end{align*}

Note that $OPT_T = \sum_{i \in T} g_L - x_i \ge \sum_{i \in T} g_L - m'$ so $g_L - m' \le \frac{OPT_T}{|T|}$.

\begin{align*}
    \tilde{g_L} - h_L &\le \frac{4}{5} \frac{\prn*{OPT_S + OPT_T + OPT_{L_r} + 2.5 OPT_T \frac{\delta n}{|T|}}}{(1 - 2\beta) |L|} \\
    &\le  \frac{4}{5} \frac{\prn*{OPT_S + OPT_T + OPT_{L_r} + 2 OPT_T \frac{\delta n}{n/2 - \delta n}}}{(1 - 2\beta) |L|} \\
    & = \frac{4}{5} \frac{\prn*{OPT_S + (1 + \Theta(\delta))OPT_T + OPT_{L_r}}}{(1 - 2\beta) |L|}. \numberthis \label{middle_bound_for_tilde_g_minus_h_L}
\end{align*}

and so we get the desired bound (due to the equivalent condition stated in \cref{ineq_we_need_to_show_to_get_delta_approx_results_foR'_L}.
\end{case}

\end{proof}

\section{Second Facility Location Proportional Mechanism Proofs}\label{sec:sec-fac-proofs}

In this section we show the strategyproofness and the expected approximation ratio of \cref{alg:second_facility_prop_mech}.

\begin{proof}[Proof of \Cref{lem:prop-is-strategyproof}]
    We follow the proof of \cite[Thm~4.1]{lu2010asymptotically}. In the proof they show that given that the first facility is at $x_k$ for any $k \in [n]$, then for any $X' = <x'_i, X_{-i}>$ that is gained from $X$ by agent $i$ deviating from $x_i$ to $x'_i$, the expected cost of agent $i$ only increases by deviating to $x'_i$. The proof does not depend at all on the location of the first facility. Since the first facility is entirely independent of the agent reported values, then the above implies that the choice of the second facility is also strategy proof.
\end{proof}

\begin{proof}[Proof of \Cref{thm:second-facility}]
  The proof is similar in structure to the one given by
  \cite{lu2010asymptotically} for the Proportional-Mechanism; the main
  difference is a careful (and tight) analysis of
  \cref{lem_bound_second_part_of_ALG_T} for the real line.

  Let $G := \{g_S, g_T\}$, 
  $S = \{ i \mid x_i \in X, d(x_i, g_S) \le d(x_i, g_T) \}$ and
  $T = [n] \setminus S$.
  
  For any $M \subseteq X$ let $OPT_M = \kmed[2](M, G)$. So $OPT_S = \kmed[1](S, g_S), OPT_T = \kmed[1](T, g_T)$ and thus
  $OPT = \kmed[2](X,G) = OPT_S + OPT_T$.

  Let $H = (g_S, h)$ be the algorithms solution for the problem, and
  $ALG = \kmed[2](X, H)$.  We denote $ALG_S, ALG_T$ to be
  $\kmed[2](S, H), \kmed[2](T, H)$ respectively, so that
  $ALG = ALG_S + ALG_T$ and therefore
  $ \E \brk{ALG} = \E [ALG_S] + \E [ALG_T] $.

  Since $ALG_S \le \kmed[1](S, g_S) = OPT_S$, we get that
  $ \E \brk*{ALG_S} \le OPT_S$, and we only need to find a bound for
  $ \E \brk*{ALG_T}$.
  \begin{align}
    \E[ALG_T] & = \sum_{i \in [n]} \E[ALG_T\mid h = x_i] Pr \prn*{h = x_i} \\
              &= \sum_{i \in T} \E[ALG_T \mid h = x_i] Pr \prn*{h = x_i} + \sum_{i \in T} \E[ALG_T \mid h = x_i] Pr \prn*{h = x_i} \label{eq_expected_ALG_T}
  \end{align}
In the summation above the first term is the expected cost due to choosing the second facility to be a location in $S$ and the second term is the expected cost due to choosing the second facility to be a location in $T$. We will bound each one individually.

For all $i \in [n]$ let $a_i = d(x_i, g_S)$, $p_i = \frac{a_i}{\sum_{j \in [n] a_j}}$ and $b_i = d(x_i, g_T)$.
So $\sum_{i \in S} a_i = OPT_S$, $\sum_{i \in T} b_i = OPT_T$ and $\sum_{i \in [n]} p_i = 1$. For each $i, j \in [n]$ let $d_{i,j} = d(x_i, x_j)$.
\begin{lem}\label{lem_bound_first_part_of_ALG_T}
    $\sum_{i \in S} \E[ALG_T \mid h = x_i] Pr \prn*{h = x_i} \le OPT_S$
\end{lem}
\begin{proof}
    \begin{align*}
        &\sum_{i \in S}
            \E[ALG_T \mid h = x_i] Pr \prn*{h = x_i} =
                \sum_{i \in S}
                    \prn*{ \sum_{t \in T} \min \crl*{a_t, d_{i,t}}}
                    \cdot p_i \\
                &\stackrel{(\star)}{\le} 
                \sum_{i \in S}
                    \prn*{ \sum_{t \in T} a_t}
                    \cdot p_i
                =
                \sum_{i \in S}
                    \prn*{\frac{\sum_{t \in T} a_t}{\sum_{t \in T} a_t + \sum_{s \in S} a_s}}
                    \cdot a_i
                \le \sum_{i \in S} a_i = OPT_S
    \end{align*}

Where inequality $(\star)$ is gained by ignoring the second facility.
\end{proof}
\begin{lem}\label{lem_bound_second_part_of_ALG_T}
    $\sum_{i \in T} \E[ALG_T \mid h = x_i] Pr \prn*{h = x_i} \le 3 OPT_T$
\end{lem}
\begin{proof}
Let $D = d(g_A, g_B)$.
\begin{align*}
    & \sum_{i \in T}
        \E[ALG_T \mid h = x_i] Pr \prn*{h = x_i}
        =
            \sum_{i \in T}
                \brk*{\prn*{ \sum_{t \in T} \min \crl*{a_t, d(x_t, x_i)}}
                \cdot p_i} \\
        &\stackrel{\textit{triangle inequality}}{\le}
            \sum_{i \in T}
                \brk*{\prn*{ \sum_{t \in T} \min \crl*{a_t, b_t + b_i}}
                \cdot p_i}
        =
            \sum_{i \in T}
                \brk*{\prn*{ \sum_{t \in T}b_t + \sum_{t \in T} \min \crl*{a_t - b_t, b_i}}
                \cdot p_i} \\
        &=
            \sum_{i \in T}
                \brk*{\prn*{ \sum_{t \in T}b_t}
                \cdot p_i}
            +
            \sum_{i \in T}
                \brk*{\prn*{ \sum_{t \in T} \min \crl*{a_t -b_t, b_i}}
                \cdot p_i} \\
        &= 
        \sum_{i \in T}
                \brk*{\prn*{ \sum_{t \in T}b_t}
                \cdot p_i}
            +
        \sum_{i \in T}
                \brk*{\prn*{ \sum_{t \in T} \min \crl*{a_t - b_t, b_i}}
                \cdot \frac{ b_i}{\sum_{j \in [n]} a_j}} \\
    & + 
        \sum_{i \in T}
                \brk*{\prn*{ \sum_{t \in T} \min \crl*{a_t -b_t, b_i}}
                \cdot \frac{a_i - b_i}{\sum_{j \in [n]} a_j}}\label{expected_alg_t_bound_one}
\end{align*}

We bound each of these 3 terms individually by $OPT_T$, and thus get the desired bound.

The first term bound:
\begin{equation}
    \sum_{i \in T}
    \brk*{\prn*{ \sum_{t \in T}b_t}
    \cdot p_i} = \prn*{ \sum_{t \in T}b_t} \sum_{i \in T} p_i \stackrel{}{\le} OPT_T
\end{equation} 
The second term bound:
\begin{align*}
    &\sum_{i \in T}
                \brk*{\prn*{ \sum_{t \in T} \min \crl*{a_t - b_t, b_i}}
                \cdot \frac{b_i}{\sum_{j \in [n]} a_j}}
    \le 
        \sum_{i \in T}
                \brk*{\prn*{ \sum_{t \in T} a_t - b_t }
                \cdot \frac{b_i}{\sum_{j \in [n]} a_j}} \\
    & \le \sum_{i \in T}
                \brk*{\prn*{ \sum_{t \in T} a_t }
                \cdot \frac{b_i}{\sum_{j \in [n]} a_j}}
    =  \sum_{i \in T} b_i = OPT_T
\end{align*}

The third term bound:
\begin{equation} \label{expression_to_bound_separately_by_space_type}
    \sum_{i \in T}
                \brk*{\prn*{ \sum_{t \in T} \min \crl*{a_t -b_t, b_i}}
                \cdot \frac{a_i - b_i}{\sum_{j \in [n]} a_j}}
    \le
        \sum_{i \in T}
                \brk*{\prn*{ \sum_{t \in T} b_i }
                \cdot \frac{a_i - b_i}{\sum_{j \in [n]} a_j}}
    =
        \sum_{i \in T}
                \brk*{ \abs{T} b_i
                \cdot \frac{a_i - b_i}{\sum_{j \in [n]} a_j}}
\end{equation}

We finish the third term bound by showing the following claim
\begin{clm} \label{clm_bound_last_part_by_OPT_T}
    \[ \sum_{i \in T}
                \brk*{ \abs{T} b_i
                \cdot \frac{a_i - b_i}{\sum_{j \in [n]} a_j}} \le OPT_T\]
\end{clm}

\begin{proof}
    
Let $T_a$ be the part of $T$ closer to $g_S$ and $T_b$ be the other part of $T$.

So $T_a = \crl*{i \mid i \in T, d(x_i, g_S) \le d(x_i,g_T)}$, and $T_b = T \setminus T_a$. Since the points are all on the line metric:

For all $i \in T_a$: $a_i = d(g_S, x_i) = d(g_S, g_T) - b_i = D - b_i$.

For all $i \in T_b$: $a_i = d(g_S, x_i) = d(g_S, g_T) + b_i = D + b_i$.

Let us use these facts in \cref{expression_to_bound_separately_by_space_type} to get:
\begin{align*}
    &\sum_{i \in T}
                    \brk*{ \abs{T} b_i
                    \cdot \frac{a_i - b_i}{\sum_{j \in [n]} a_j}}
    =
        \sum_{i \in T_a}
                    \brk*{ \abs{T} b_i
                    \cdot \frac{D - b_i - b_i}{\sum_{j \in [n]} a_j}}
        +
        \sum_{i \in T_b}
                    \brk*{ \abs{T} b_i
                    \cdot \frac{D + b_i - b_i}{\sum_{j \in [n]} a_j}} \\
    &= \frac{\abs{T} D}{\sum_{j \in [n]} a_j} \prn*{ \sum_{i \in T} b_i } - \frac{2 \abs{T} \sum_{i \in T_a} {b_i}^2}{\sum_{j \in [n]} a_j}
    = \frac{\abs{T} D \cdot OPT_T - 2 \abs{T} \sum_{i \in T_a} {b_i}^2}{\sum_{j \in [n]} a_j} \\
    & \stackrel{(\star)}{\le} \frac{\abs{T} D \cdot OPT_T - 2 (OPT_{T_a})^2}{\sum_{j \in [n]} a_j} \stackrel{(\star \star)}{\le} \frac{\abs{T} D - 2 \frac{(OPT_{T_a})^2}{OPT_T}}{\abs{T} D + OPT_{T_b} - OPT_{T_a}} \cdot OPT_T \numberthis \label{ineq:second_fac_last_term}
\end{align*}

$(\star)$ is due to $QM-AM$ inequality since 
\[
    \sum_{i \in T_a} {b_i}^2 \ge \frac{\prn*{\sum_{i \in T_a} b_i}^2}{\abs{T_a}} \stackrel{\abs{T_a} \le \abs{T}}{\ge} (OPT_{T_a})^2 \frac{1}{\abs{T}}.
\]

$(\star \star)$ explanation: \begin{align*}
    &\sum_{j \in [n]} a_j = \sum_{j \in S} a_j + \sum_{j \in T_a} a_j + \sum_{j \in T_b} a_j =  OPT_S + \sum_{j \in T_a} D - b_j + \sum_{j \in T_b} D + b_j \\
    &= OPT_S + D\abs{T} + OPT_{T_b} - OPT_{T_a} \ge \abs{T} D + OPT_{T_b} - OPT_{T_a} 
\end{align*}

We now show that \[
    \frac{\abs{T} D - 2 \frac{(OPT_{T_a})^2}{OPT_T}}{\abs{T} D + OPT_{T_b} - OPT_{T_a}} \le 1.
\]

Indeed, if $OPT_{T_b} \ge OPT_{T_a}$ then $\abs{T}D + OPT_{T_b} - OPT_{T_a} \ge \abs{T} D$ and thus:
\[
    \frac{\abs{T} D - 2\frac{(OPT_{T_a})^2}{OPT_T}}{\abs{T} D + OPT_{T_b} - OPT_{T_a}} \le \frac{\abs{T} D - 2\frac{(OPT_{T_a})^2}{OPT_T}}{\abs{T} D} \le \frac{\abs{T}D}{\abs{T}D} = 1.
\]

Otherwise $OPT_{T_a} - OPT_{T_b} > 0$ and thus:
\begin{align*}
    &OPT_{T_a} > OPT_{T_b} \implies \\
    &OPT_T = OPT_{T_a} + OPT_{T_b} < 2 OPT_{T_a} \implies \\
    &(OPT_{T_a} - OPT_{T_b}) OPT_T \le 2 (OPT_{T_a} - OPT_{T_b}) OPT_{T_a} = 2 (OPT_{T_a})^2 - 2 OPT_{T_b} OPT_{T_a} \le 2 \prn{OPT_{T_a}}^2 \implies \\
    &OPT_{T_a} - OPT_{T_b} \le \frac{2 \prn{OPT_{T_a}}^2}{OPT_T} \implies 
    \frac{\abs{T} D - 2\frac{(OPT_{T_a})^2}{OPT_T}}{\abs{T} D + OPT_{T_b} - OPT_{T_a}} \le \frac{{\abs{T} D + OPT_{T_b} - OPT_{T_a}}}{{\abs{T} D + OPT_{T_b} - OPT_{T_a}}} = 1
\end{align*}

So from \cref{ineq:second_fac_last_term}: $\sum_{i \in T} \brk*{ \abs{T} b_i \cdot \frac{a_i - b_i}{\sum_{j \in [n]} a_j}} \le OPT_T$.

\end{proof}

To summarize:
$\sum_{i \in T} \E[ALG_T \mid h = x_i] Pr \prn*{h = x_i} \le OPT_T + OPT_T + OPT_T = 3 OPT_T$.
\end{proof}

From plugging \cref{lem_bound_first_part_of_ALG_T}, \cref{lem_bound_second_part_of_ALG_T} into \cref{eq_expected_ALG_T} we get:

\[
    \E \brk*{ALG} \le 2 OPT_S + 3 OPT_T \le 3 OPT.
\]
\end{proof}
\end{document}